\documentclass[11pt,journal,draftcls,a4paper,onecolumn]{IEEEtran}

\usepackage{amsmath}
\usepackage{amsthm}
\usepackage{amssymb}
\usepackage{mathrsfs} 
\usepackage{siunitx} 
\usepackage{mathtools}

\usepackage{stmaryrd}  

\usepackage{enumitem} 

\usepackage[ruled,vlined]{algorithm2e} 

\usepackage{booktabs}
\usepackage[caption=false,font=footnotesize]{subfig} 
\usepackage{multirow} 
\usepackage[figuresleft]{rotating} 

\usepackage{booktabs}
\usepackage[caption=false,font=footnotesize]{subfig} 
\usepackage{multirow} 
\usepackage[figuresleft]{rotating} 

\usepackage[noadjust]{cite}
\usepackage{graphicx}
\DeclareGraphicsExtensions{.pdf,.jpg,.png}
\usepackage{textcomp} 
\usepackage{xspace} 

\usepackage{url}
\usepackage{hyperref}  
\hypersetup{colorlinks,%
            citecolor=red,%
            filecolor=black,%
            linkcolor=blue,%
            urlcolor=blue}

\def\Y{\mathbf{Y}}

\def\M{\mathbf{M}}
\def\A{\mathbf{A}}

\def\dM{\mathbf{dM}}

\def\I{\mathbf{I}}
\def\Z{\mathbf{Z}}
\def\shuffle#1{\mathbf{S}_{#1}}

\def\Mtilde{\tilde{\mathbf{M}}}

\def\Hf{\mathbf{H}_{(\A,\dM)}f}

\def\Mt{\widetilde{\M}}
\def\gt{\hat{g}_t}
\def\Ne{N_{\text{epochs}}}

\def\Nd{N_{\text{iter}}^D}
\def\Np{N_{\text{iter}}^P}
\def\Ni{N_{\text{iter}}}

\newcommand{\wrt}{with respect to\xspace}

\def\Ball#1#2{\mathcal{B}_{\text{F}}(#1,#2)}
\def\Bs{\mathcal{B}_{\text{F}}(\mathbf{0},\sigma)}
\def\Sk{\mathcal{S}_{\nendm}}
\def\Ak{\mathcal{A}_{\nendm}}
\def\Bk#1{\mathcal{B}_{\text{F}}(#1,\kappa)}
\def\Qt{\mathcal{Q}(\Y_t,\M)}

\def\D{\mathcal{D}}

\newcommand{\Bf}[2]{\mathcal{B}_{\text{F}}(#1,#2)}
\newcommand{\R}[2]{\mathbb{R}^{#1 \times #2}}
\newcommand{\Zero}[2]{\mathbf{0}_{#1,#2}}

\def\1#1{\mathbf{1}_{#1}}

\def\normf#1{\left\lVert #1 \right\rVert_{\text{F}}}%
\def\normF2#1{\left\lVert #1 \right\rVert_{\text{F}}^2}%
\def\norm2#1{\lVert #1 \rVert_{2}^2}%
\def\argmin#1{\underset{#1}{\arg \min \,}}

\def\t#1{#1^{\text{T}}}

\DeclareMathOperator{\aSAM}{aSAM}
\DeclareMathOperator{\GMSE}{GMSE}
\DeclareMathOperator{\RE}{RE}

\DeclareMathOperator{\Tr}{Tr}

\newcommand\nbpix{N}
\newcommand\nendm{R}
\newcommand\nband{L}
\newcommand\ntime{T}

\newcommand{\second}[1]{\textcolor{blue}{#1}}

\newtheorem{lemma}{Lemma}
\newtheorem{proposition}{Proposition}

\theoremstyle{definition}

\newtheorem*{remark}{Remark}
\newtheorem{assumption}{Assumption}

\usepackage[dvipsnames]{xcolor}

\newcolumntype{L}[1]{>{\raggedright\let\newline\\\arraybackslash\hspace{0pt}}m{#1}}
\newcolumntype{C}[1]{>{\centering\let\newline\\\arraybackslash\hspace{0pt}}m{#1}}
\newcolumntype{R}[1]{>{\raggedleft\let\newline\\\arraybackslash\hspace{0pt}}m{#1}}

\title{Online Unmixing of Multitemporal Hyperspectral Images accounting for Spectral Variability}

\author{Pierre-Antoine Thouvenin,~\IEEEmembership{Student Member,~IEEE}, Nicolas Dobigeon,~\IEEEmembership{Senior Member,~IEEE} and Jean-Yves~Tourneret,~\IEEEmembership{Senior Member,~IEEE}

\thanks{This work is supported by the Direction G\'en\'erale de l'Armement, French Ministry of Defence.}
\thanks{The authors are with the University of Toulouse, IRIT/INP-ENSEEIHT, 2 rue Camichel, BP 7122, 31071 Toulouse cedex 7, France. (e-mail: \{pierreantoine.thouvenin, Nicolas.Dobigeon, Jean-Yves.Tourneret\}@enseeiht.fr}}

\begin{document}
\maketitle
\begin{abstract}
Hyperspectral unmixing is aimed at identifying the reference spectral signatures composing an hyperspectral image and their relative abundance fractions in each pixel. In practice, the identified signatures may vary spectrally from an image to another due to varying acquisition conditions, thus inducing possibly significant estimation errors. Against this background, hyperspectral unmixing of several images acquired over the same area is of considerable interest. Indeed, such an analysis enables the endmembers of the scene to be tracked and the corresponding endmember variability to be characterized. Sequential endmember estimation from a set of hyperspectral images is expected to provide improved performance when compared to methods analyzing the images independently. However, the significant size of hyperspectral data precludes the use of batch procedures to jointly estimate the mixture parameters of a sequence of hyperspectral images. Provided that each elementary component is present in at least one image of the sequence, we propose to perform an online hyperspectral unmixing accounting for temporal endmember variability. The online hyperspectral unmixing is formulated as a two-stage stochastic program, which can be solved using a stochastic approximation. The performance of the proposed method is evaluated on synthetic and real data. A comparison with independent unmixing algorithms finally illustrates the interest of the proposed strategy.
\end{abstract}
\begin{IEEEkeywords}
Hyperspectral imagery, perturbed linear unmixing (PLMM), endmember temporal variability, two-stage stochastic program, stochastic approximation (SA). 
\end{IEEEkeywords}

\section{Introduction} \label{sec:intro}
\IEEEPARstart{H}{yperspectral} imagery has known an increasing interest over the past decades due to the significant spectral information it conveys. Acquired in hundreds of contiguous spectral bands (e.g., from \SIrange{300}{2600}{\nano \metre} for the AVIRIS sensor), hyperspectral (HS) images facilitate the identification of the elements composing the imaged scene\footnote{Note that some airborne sensors cover larger wavelength range, while some mounted sensors can offer millimeter spatial resolutions.}. However, the high spectral resolution of these images is mitigated by their lower spatial resolution, which results in pixel spectra composed of mixtures of reference signatures. Spectral unmixing consists of determining the reference spectral signatures composing the data -- referred to as \emph{endmembers} -- and their abundance fractions in each pixel according to a predefined mixture model accounting for several environmental factors (declivity, multiple reflections, ...). Provided microscopic interactions between the materials of the imaged scene are negligible and the relief of the scene is flat, a linear mixing model (LMM) is traditionally used to describe the data \cite{Bioucas2012jstars}. 
However, varying acquisition conditions such as illumination or natural evolution of the scene may significantly alter the shape and the amplitude of the spectral signatures acquired, thus affecting the extracted endmembers from an image to another. In this context, HS unmixing of several images acquired over the same area at different time instants can be of considerable interest. Indeed, such an analysis enables the endmembers of the scene and endmember variability to be assessed, thus improving endmember estimation when compared to independent image analyses performed with any state-of-the-art unmixing method.

So far, spatial variability within a given image has been considered in various models either derived from a statistical or a deterministic point of view \cite{Zare2014IEEESPMAG}. The first class of methods assumes that the endmember spectra are realizations of multivariate distributions \cite{Eches2010ip,Du2014,Halimi2014}. The second class of methods represents endmember signatures as members of spectral libraries associated with each material (bundles) \cite{Somers2012jstars2}. Another recently proposed approach consists in estimating the parameters of an explicit variability model \cite{Thouvenin2015}. To the best of our knowledge, spatio-temporal variability has been analyzed for the first time in the Bayesian framework proposed in \cite{Halimi2015eusipco}. Another recent contribution similarly resorts to a batch estimation technique to address spectral unmixing of multi-temporal HS images \cite{Henrot2015}. However, HS unmixing using a significant number of images or several large images precludes the use of batch estimation procedures as in \cite{Halimi2015eusipco,Henrot2015} due to limited memory and computational resources. Since online estimation procedures enable data to be sequentially incorporated into the estimation process without the need to simultaneously load all the data into memory, we focus in this paper on the design of an online HS unmixing method accounting for temporal variability.

Since the identified endmembers can be considered as time-varying instances of reference endmembers, we use the perturbed linear mixing model (PLMM) proposed in \cite{Thouvenin2015} to account for spectral variability. However, inspired by the works presented in \cite{Ralph2011,Mairal2010}, we formulate the unmixing problem as a two-stage stochastic program that allows the model parameters to be estimated online contrary to the algorithm proposed in \cite{Thouvenin2015}. To the best of our knowledge, it is the first time HS unmixing accounting for temporal variability has been formulated as a two-stage stochastic program solved by an online\footnote{The terminology ``online'' is slightly abusive in our context since the time difference between two consecutive images can extend to several months.} algorithm.

The paper is organized as follows. The proposed PLMM accounting for temporal variability is introduced in Section \ref{sec:PLMM}. Section \ref{sec:TSSP} describes an online algorithm to solve the resulting optimization problem. Experimental results obtained on synthetic and real data are reported in Sections \ref{sec:simulations} and \ref{sec:experiments} respectively. The results obtained with the proposed algorithm are systematically compared to those obtained with the vertex component analysis / fully constrained least squares (VCA \cite{Nascimento2005} / FCLS \cite{Heinz2001,Bioucas2010}), SISAL \cite{Bioucas2009} / FCLS, the $\ell_{1/2}$ non-negative matrix factorization (NMF) \cite{Qian2011} and the BCD/ADMM algorithm of \cite{Thouvenin2015}, each method being independently applied to each image of the sequence. Section \ref{sec:conclusion} finally concludes this work.
%

\begin{table}[t!]
\centering
\caption{Notations.}
\begin{tabular}{ll}
\hline
$\nbpix$  & number of pixels \\
$\nband $ & number of spectral bands \\
$\nendm $ & number of endmembers \\
$\ntime $ & number of images \\
$\mathbf{y}_{nt} \in \mathbb{R}^\nband$ & $n$th pixel of the $t$th image \\
$\Y_t  \in \R{\nband}{\nbpix}$ & lexicographically ordered pixels of the $t$th image\\
$\M    \in \R{\nband}{\nendm }$ & endmember matrix\\
$\dM_t \in \R{\nband}{\nendm }$ & $t$th variability matrix\\
$\A_t    \in \R{\nendm}{\nbpix}$& $t$th abundance matrix\\
$\mathbf{a}_{nt} \in \mathbb{R}^\nendm$ & $n$th column of the matrix $\A_t$\\
$\succeq$ & component-wise inequality \\
$\mathcal{Y}$ & $[0,1]^{\nband \times \nbpix}$ \\
$\mathcal{M}$ & $[0,1]^{\nband \times \nendm}$ \\
$\Sk$ & unit simplex of $\mathbf{R}^\nendm$ \\
$\Ak$ & $ \left\{ \A \in \R{\nendm}{\nbpix} \middle| \mathbf{a}_n \in \Sk, \, \forall n \in \llbracket 1,\nbpix \rrbracket \right\}$ \smallskip\\
$\Bk{\Z}$ & $\left\{ \mathbf{X} \in \R{\nband}{\nendm} \middle| \normf{\mathbf{X} - \Z} \leq \kappa \right\}$, $\kappa > 0$ \smallskip \\
$\mathcal{D}$ & $\Bs \cap \left\{\dM \middle| \normf{\mathbb{E}[\dM]} \leq \kappa \right\}$ \smallskip \\
$\mathcal{D}_t$ & $\Bs \cap \left\{\dM \middle| \normf{ \sum_{i=1}^{t-1} \dM_i + \dM} \leq t\kappa \right\} $ \\
$\mathcal{Z}_t$ & $\Ak \times \mathcal{D}_t$ \\
$\Qt$ &$\{(\A,\dM) \in \mathcal{Z}_t | \nabla_{(\A,\dM)} f(\Y_t,\M,\A,\dM) = \mathbf{0} \}$ \\
$\mathcal{P}_{\mathcal{S}}$ & projector on the set $\mathcal{S}$ \\
$\mathcal{P}_+$ & projector on $\left\{ \mathbf{X} \in \R{\nband}{\nendm} \middle| \mathbf{X} \succeq \Zero{\nband}{\nendm} \right\}$ \\
$\langle \mathbf{X}, \mathbf{Y} \rangle$ & matrix inner product $\Tr(\t{\mathbf{X}} \mathbf{Y})$ \\
$\iota_\mathcal{S} (x)$ & $\left\{ \begin{array}{ll}
0 &\text{ if } x \in \mathcal{S} \\
+\infty &\text{ otherwise.}
\end{array} \right.$ \\
\hline \vspace{-0.5cm}
\end{tabular}
\end{table}

\section{Problem statement} \label{sec:PLMM}

	\subsection{Perturbed linear mixing model (PLMM)} \label{p:PLMM}
We consider HS images acquired at $\ntime$ different time instants over the same scene, assuming that at most $\nendm$ endmembers are present in the resulting time series and that the images share these $\nendm$ common endmembers. Each endmember does not need to be present in each image, but at least in one image of the time series. Given an \emph{a priori} known number of endmembers $\nendm$, the PLMM consists in representing each pixel $\mathbf{y}_{nt}$ by a linear combination of the $\nendm$ endmembers -- denoted by $\mathbf{m}_r$ -- affected by a perturbation vector $\mathbf{dm}_{rt}$ accounting for temporal endmember variability. The proposed model considers the case where the variability essentially results from the evolution of the scene or from the global acquisition conditions from one image to another. As a first approximation, the variability is assumed to be constant on each image. The resulting PLMM can thus be written \vspace{-0.2cm}
\begin{equation}
\label{eq:model}
\mathbf{y}_{nt}  =  \sum_{r=1}^\nendm a_{rnt}\Bigl(\mathbf{m}_r + \mathbf{dm}_{rt} \Bigr) + \mathbf{b}_{nt} \vspace{-0.2cm}
\end{equation}%
for $n = 1,\dotsc , \nbpix$ and $t = 1, \dotsc,\ntime$, where $\mathbf{y}_{nt}$ denotes the $n$th image pixel at time $t$, $\mathbf{m}_r$ is the $r$th endmember, $a_{rnt}$ is the proportion of the $r$th endmember in the $n$th pixel at time $t$, and $\mathbf{dm}_{rt}$ denotes the perturbation of the $r$th endmember at time $t$. Finally, $\mathbf{b}_{nt}$ models the noise resulting from the data acquisition and the modeling errors. In matrix form, the PLMM \eqref{eq:model} can be written as \vspace{-0.2cm}
\begin{equation}
\Y_t  = (\M+\dM_t)\A_t + \mathbf{B}_t \vspace{-0.2cm}
\end{equation}%
where $\Y_t = \left[ \mathbf{y}_{1t},\dotsc,\mathbf{y}_{Nt} \right]$ is an $\nband \times \nbpix$ matrix containing the pixels of the $t$th image, $\M$ denotes an $\nband \times \nendm$ matrix containing the endmembers, $\A_t$ is an $\nendm \times \nbpix$ matrix composed of the abundance vectors $\mathbf{a}_{nt}$, $\dM_t$ is an $\nband \times \nendm$ matrix whose columns are the perturbation vectors associated with the $t$th image, and $\mathbf{B}_t$ is an $\nband \times \nbpix$ matrix accounting for the noise at time instant $t$. The non-negativity and sum-to-one constraints usually considered to reflect physical considerations are \vspace{-0.2cm}
\begin{align} \label{eq:constraints}
\begin{split}
\A_t & \succeq \mathbf{0}_{\nendm ,\nbpix}, \quad  \A_t^T \mathbf{1}_\nendm  = \mathbf{1}_\nbpix, \, \forall t = 1,\dotsc , \ntime   \\%
\M & \succeq \mathbf{0}_{\nband ,\nendm }
\end{split} \vspace{-0.2cm}
\end{align}{}%
where $\succeq$ denotes a component-wise inequality. We also consider the following assumptions on the inherent variability of the observed scenes \vspace{-0.2cm}
\begin{align} 
\normF2{\dM_t} &\leq \sigma^2, \quad \text{for} \ t = 1,\dotsc,\ntime \label{eq:bound_dM} \\
\normF2{\frac{1}{\ntime} \sum_{t=1}^\ntime \dM_t} & \leq \kappa^2 \label{eq:bound_mean_dM} \vspace{-0.2cm}
\end{align}
where $\sigma$ and $\kappa$ are fixed positive constants, and $\normf{\cdot}$ denotes the Frobenius norm.
 These two constraints can be interpreted in terms of the feasible domain of $\M$ and $\dM_t$.
Indeed, introducing the perturbed endmembers $\M_t \triangleq \M + \dM_t$, the constraint \eqref{eq:bound_dM} can be reformulated as \vspace{-0.2cm}
\begin{align*}
\normF2{\dM_t} = \normF2{\M - \M_t} \leq \sigma^2 \Leftrightarrow \M \in \bigcap_{t=1}^T \Bf{\M_t}{\sigma} \vspace{-0.2cm}
\end{align*}
where $\Bf{\M_t}{\sigma}$ is the ball of center $\M_t$ and of radius $\sigma$. This highlights the fact that the number of constraints imposed on the endmembers increases with $\ntime$, i.e., the more images are processed, the more information can be extracted in terms of endmember signatures. On the other hand, \eqref{eq:bound_mean_dM} constrains the perturbed endmembers to be distributed around the true endmembers, i.e., the endmember signatures $\M$ should reflect the average behavior of the perturbed endmembers $\M_t$ in the sequence. In practice, setting $\sigma$ to a reasonable value is desirable from a modeling point of view, since very large perturbations should probably be interpreted as outliers, thus leading to the removal of the corrupted elements from the unmixing process. Note however that the algorithm proposed in Section \ref{p:parameter_estimation} is independent from any consideration on the values of $\sigma^2$ and $\kappa^2$.

\begin{remark}
In practice, HS unmixing is performed on reflectance data, hence $\Y_t \in [0,1]^{\nband \times \nbpix}$. The abundance sum-to-one and non-negativity constraints further imply $\M~\in~[0,1]^{\nband \times \nendm}$. In fact, the compactness of both the data support and the space associated with the endmember constraints -- denoted by $\mathcal{Y}$ and $\mathcal{M}$ respectively -- is crucial for the convergence result given in Paragraph \ref{p:cv_proof}. In addition, the images $\Y_t$ can be assumed to be independent and identically distributed (i.i.d.) since these images have been acquired by possibly different sensors at different time instants.
\end{remark}

\subsection{Problem formulation} \label{subsec:problem_formulation}

In order to design an online estimation algorithm, the model \eqref{eq:model} combined with the constraints \eqref{eq:constraints} can be used to formulate a two-stage stochastic program consisting in estimating the endmembers present in the image sequence. Since only the endmembers are supposed to be commonly shared by the different images, we propose to minimize a marginal cost function obtained by marginalizing an instantaneous cost function over the abundances and the variability terms, so that the resulting cost only depends on the endmembers. Assuming the expectations are well-defined, we consider the following optimization problem \vspace{-0.2cm}
\begin{equation}
\label{eq:problem}
\min_{\M \in \mathcal{M}} g(\M) = \mathbb{E}_{\Y,\A,\dM} \Bigl[ f \bigl(\Y,\M,\A,\dM \bigr) \Bigr] \vspace{-0.2cm}
\end{equation}
where $\mathcal{M} = [0,1]^{\nband \times \nendm}$ and where the function $f$ is defined as \vspace{-0.2cm}
\begin{equation}
\label{eq:f}
\begin{split}
f(\Y,\M,\A,\dM) = &\frac{1}{2} \normF2{\Y - (\M+\dM)\A} \\
&+ \alpha \Phi(\A) + \beta \Psi(\M) + \gamma \Upsilon(\dM)
\end{split}. \vspace{-0.2cm}
\end{equation}
$\Phi,\Psi$ and $\Upsilon$ denote appropriate penalization terms on the abundances, the endmembers and the variability with \vspace{-0.2cm}
\begin{align}
\A \in \Ak &= \left\{ \A \in \R{\nendm}{\nbpix} \middle| \, \mathbf{a}_n \in \Sk, \, \text{for} \: n = 1,\dotsc,\nbpix \right\} \\
\dM \in \mathcal{D} &= \Bs \cap \left\{\dM \; \middle| \; \normf{\mathbb{E} \bigl[ \dM \bigr]} \leq \kappa \right\}. \vspace{-0.2cm}
\end{align}
 The parameters $\alpha,\beta$ and $\gamma$ ensure a trade-off between the data fitting term and the penalties. In practice, $g$ is approximated at time $t$ by an upper bound $\gt$ given by a stochastic approximation \cite{Mairal2010} \vspace{-0.2cm}
\begin{equation}
\label{eq:g_hat}
\begin{split}
&\gt (\M) = \frac{1}{2t} \sum_{i=1}^t \normF2{ \Y_i - (\M + \dM_i)\A_i} + \beta \Psi(\M) \\
&= \frac{1}{t} \sum_{i=1}^t \left( \frac{1}{2} \normF2{\M\A_i} - \langle \Y_i - \dM_i\A_i, \M\A_i \rangle \right) \\
& + \beta \Psi(\M) + c \\
& = \frac{1}{t} \left[ \frac{1}{2} \Tr( \t{\M} \M \mathbf{C}_t) + \Tr(\t{\M}\mathbf{D}_t) \right] + \beta \Psi(\M) + c
\end{split} \vspace{-0.2cm}
\end{equation}
where $\langle \mathbf{X}, \mathbf{Y} \rangle = \Tr(\t{\mathbf{X}} \mathbf{Y})$, $c$ is a constant independent from $\M$ and \vspace{-0.2cm}
\begin{equation}
\label{eq:CD}
\mathbf{C}_t = \sum_{i=1}^t \A_i\t{\A_i}, \quad
\mathbf{D}_t = \sum_{i=1}^t (\dM_i\A_i - \Y_i)\t{\A_i}. \vspace{-0.2cm}
\end{equation}
Besides, $\mathcal{D}$ is approximated by \vspace{-0.2cm}
\begin{equation}
\mathcal{D}_t = \Bs \cap \left\{ \dM \; \middle| \; \normf{\dM + \mathbf{E}_{t-1}} \leq t\kappa \right\} 
\end{equation}
with \vspace{-0.2cm}
\begin{equation} \label{eq:Et}
\mathbf{E}_t = \sum_{i=1}^t \dM_i. \vspace{-0.2cm}
\end{equation}
Examples of penalizations that will be considered in this study are detailed in the following paragraphs.

		\subsubsection{Abundance penalization} \label{subsec:A_penalization}
		
In this work, the abundance penalization $\Phi$ has been chosen to promote temporally smooth abundances -- in the $\ell_2$-norm sense -- between two consecutive images, leading to \vspace{-0.2cm}
\begin{equation}
\Phi(\A_t) = \frac{1}{2}\normF2{\A_t - \A_{t-1}}. \vspace{-0.15cm}
\end{equation}
As long as $\Phi$ satisfies the regularity condition given in Paragraph \ref{p:cv_proof}, any other type of prior knowledge relative to the abundances can be incorporated into the proposed method.

		\subsubsection{Endmember penalization}
\label{subsec:M_penalization}
Classical endmember penalizations found in the literature consist in constraining the size of the $(\nendm-1)$-simplex whose vertices are the endmembers. In this paper, we consider the mutual distance between each endmember introduced in \cite{Berman2004,Arngren2011}, defined as \vspace{-0.2cm}
\begin{equation} \label{eq:reg_M}
\begin{split}
\Psi(\mathbf{M}) = \frac{1}{2}\sum_{i = 1}^\nendm  \Biggl( \underset{j \neq i}{\sum_{j = 1}^\nendm } \norm2{\mathbf{m}_i - \mathbf{m}_j} \Biggr) 
= \frac{1}{2} \sum_{r = 1}^\nendm  \normF2{\mathbf{MG}_r}
\end{split} \vspace{-0.2cm}
\end{equation}
where \vspace{-0.2cm}
\begin{equation}
\mathbf{G}_r = - \mathbf{I}_\nendm + \mathbf{e}_r \mathbf{1}_\nendm^T \vspace{-0.2cm}
\end{equation}
and $\mathbf{e}_r$ denotes the $r$th canonical basis vector of $\mathbb{R}^\nendm$.

		\subsubsection{Variability penalization} \label{subsec:dM_penalization}
Assuming that the spectral variation between two consecutive images is \emph{a priori} temporally smooth, we consider the following $\ell_2$-norm penalization \vspace{-0.2cm}
\begin{equation}
\Upsilon (\dM_t) = \frac{1}{2} \normF2{\dM_t - \dM_{t-1}}. \vspace{-0.2cm}
\end{equation}
Similarly, any other type of prior knowledge relative to the variability can be considered as long as $\Upsilon$ satisfies the regularity condition given in Paragraph \ref{p:cv_proof}.

\begin{algorithm}[!t]
\footnotesize{
 \KwData{$\M^{(0)}$, $\A_0$, $\dM_0$, $\alpha > 0$, $\beta > 0$, $\gamma > 0$, $\xi \in ]0,1]$}
 \Begin{
  $\mathbf{C}_0 \leftarrow \Zero{\nendm}{\nendm}$\;
  $\mathbf{D}_0 \leftarrow \Zero{\nband}{\nendm}$\;
  $\mathbf{E}_0 \leftarrow \Zero{\nband}{\nendm}$\;
  \For{$t = 1$ \KwTo $\ntime$}{
    \lnlset{select_Y}{a} Random selection of an image $\Y_t$ (random permutation of the image sequence)\;
    \tcp{Abundance and variability estimation by PALM \cite{Bolte2013}, cf. \S \ref{pp:update_A_dM}}
	\lnlset{PALM_A_dM}{b}$(\A_t, \dM_t) \in \argmin{(\A,\dM) \in \Ak \times \D_t} f(\Y_t,\M^{(t)},\A,\dM)$\;
	$\mathbf{C}_t \leftarrow \xi\mathbf{C}_{t-1} + \A_t\t{\A_t}$\; \smallskip
    $\mathbf{D}_t \leftarrow \xi\mathbf{D}_{t-1} + (\dM_t \A_t - \Y_t )\t{\A_t}$\; \smallskip
    $\mathbf{E}_t \leftarrow \xi\mathbf{E}_{t-1} + \dM_t$\; \medskip
    \tcp{Endmember update \cite[Algo. 2]{Mairal2010}, cf. \S \ref{pp:update_M}}
    \lnlset{PALM_M}{c}$\M^{(t)} \leftarrow \argmin{\M \in \mathcal{M}} \gt(\M)$\;
   }
 }
 \KwResult{$\M^{(\ntime)}$, $(\A_t)_{t = 1,\cdots, \ntime}$, $(\dM_t)_{t = 1,\cdots, \ntime}$}
}
 \caption{Online unmixing algorithm.}
 \label{alg:SAA} 
\end{algorithm}

\section{A two-stage stochastic program} \label{sec:TSSP}

	\subsection{Two-stage stochastic program: general principle}
The following lines briefly recall the main ideas presented in the introduction of \cite{Ralph2011}. A two-stage stochastic program is generally expressed as \vspace{-0.1cm}
\begin{equation}
\label{eq:general_problem}
\min_{\M} \mathbb{E}_{\Y,\Z} \Bigl[ f \bigl( \Y,\M,\Z \bigr) \Bigr] \; \text{s.t.} \; \M \in \mathcal{M}, \text{ with} \; \Z \in \mathcal{Z}. \vspace{-0.1cm}
\end{equation}
At the first stage, $\M$ must be chosen before any new data $\Y$ is available. At the second-stage, when $\M$ has been fixed and a new data is acquired, the second-stage variable $\Z$ is computed as the solution (if it is unique and well defined) to the optimization problem \vspace{-0.2cm}
\begin{equation}
\label{eq:second_stage}
\min_{\Z \in \mathcal{Z}} f(\Y,\M,\Z). \vspace{-0.2cm}
\end{equation}
Given an independent and identically distributed (i.i.d) $\ntime$-sample $(\Y_1,\dotsc,\Y_\ntime)$, problem \eqref{eq:general_problem} can be approximated by the sample average approximation (SAA) \vspace{-0.2cm}
\begin{equation}
\label{eq:SAA1}
\min_{\M,\Z_1,\dotsc, \Z_\ntime} \frac{1}{\ntime} \sum_{t=1}^T f(\Y_t,\M,\Z_t), \, \text{s.t.} \, \M \in \mathcal{M}, \, \Z_t \in \mathcal{Z}. \vspace{-0.2cm}
\end{equation}
Moreover, when the second-stage \eqref{eq:second_stage} admits a unique solution, \eqref{eq:SAA1}  can be rewritten as \vspace{-0.2cm}
\begin{align}
\label{eq:SAA2}
&\min_{\M \in \mathcal{M}} \frac{1}{\ntime} \sum_{t=1}^T h(\Y_t,\M) \\
&h(\Y_t,\M) = \min_{\Z \in \mathcal{Z}} f(\Y_t,\M,\Z) \label{eq:prob_A_dM} \vspace{-0.2cm}
\end{align}
which is the SAA corresponding to \vspace{-0.2cm}
\begin{align}
& \min_{\M \in \mathcal{M}} \mathbb{E}_\Y \bigl[ h(\Y,\M)  \bigr] \\
& h(\Y,\M) = \min_{\Z \in \mathcal{Z}} f \bigl(\Y,\M,\Z \bigr) \vspace{-0.2cm}
\end{align}
where the two stages explicitly appear. However, $f$ defined in \eqref{eq:f} is non-convex \wrt{} $\Z = (\A,\dM)$, where $\mathcal{Z} = \Ak \times \mathcal{D}$. Thus, problem \eqref{eq:second_stage} does not admit a unique global minimum, and existing algorithms will at most provide a critical point of $f (\Y,\M,\cdot) + \iota_\mathcal{Z}$, where $\iota_\mathcal{\mathcal{Z}}$ denotes the indicator function of the set $\mathcal{Z}$. In this specific case, a new convergence framework based on a generalized equation has been developed in \cite{Ralph2011}. 
Such a framework enables a convergence result in terms of a critical point $\{\M,\Z_1,\dotsc, \Z_\ntime \}$ of \eqref{eq:SAA1} to be obtained. However, the significant size of the SAA problem \eqref{eq:SAA1} in our case is generally too expensive from a computational point of view. To alleviate this problem, we propose to slightly adapt the work developed in \cite{Mairal2010} to propose an online estimation algorithm described in Algo. \ref{alg:SAA}. This algorithm has the same convergence property as \cite{Mairal2010} provided the non-convex function $f (\Y,\M,\cdot) + \iota_\mathcal{Z}$ exclusively admits locally unique critical points. Further details are given in Paragraph \ref{p:cv_proof}.

\begin{algorithm}[!t]
\footnotesize{
 \KwData{$\Y_t,\M^{(t)},\A_t^{(0)},\dM_t^{(0)},\mathbf{E}_{t-1}$}
 \Begin{
  $k \leftarrow 0$\;
  \While{stopping criterion not satisfied}{
  	\tcp{Abundance update}
    $\A_t^{(k+1)} \leftarrow \text{Update}\left( \A_t^{(k)} \right)$\tcp*[r]{cf. \eqref{eq:update_A}}
     
    \tcp{Variability update}
    $\dM_t^{(k+1)} \leftarrow \text{Update} \left( \dM_t^{(k)} \right)$\tcp*[r]{cf. \eqref{eq:update_dM}}
    
	$k \leftarrow k+1$\;
   }
   $\A_t \leftarrow \A_t^{(k)}$, $\dM_t \leftarrow \dM_t^{(k)}$\;
 }
  \KwResult{$(\A_t,\dM_t)$}
 }
 \caption{\footnotesize{Abundance and variability estimation using PALM.}}
 \label{alg:update_A_dM}
\end{algorithm}

\begin{algorithm}[!ht]
\footnotesize{
 \KwData{$\M^{(t,0)} = \M^{(t-1)},\mathbf{C}_t,\mathbf{D}_t$}
 \Begin{
  $k \leftarrow 0$\;
  \While{stopping criterion not satisfied}{
  	\tcp{Endmember update}
    $\M^{(t,k+1)} \leftarrow \text{Update}\left( \M^{(t,k)} \right)$\tcp*[r]{cf. \eqref{eq:update_M}}
	$k \leftarrow k+1$\;
   }
   $\M^{(t)} \leftarrow \M^{(t,k)}$\;
 }
  \KwResult{$\M^{(t)}$}
 }
 \caption{\footnotesize{Endmember estimation.}}
 \label{alg:update_M}
\end{algorithm}
	\subsection{Parameter estimation} \label{p:parameter_estimation}
	
Whenever an image $\Y_t$ has been received, the abundances and variability are estimated by a proximal alternating linearized minimization (PALM) algorithm \cite{Bolte2013}, which is guaranteed to converge to a critical point $(\A^*,\dM^*)$ of $f (\Y_t,\M,\cdot,\cdot) + \iota_{\mathcal{\Ak}\times \mathcal{D}_t}$. The endmembers are then updated by proximal gradient descent steps, similarly to \cite{Mairal2010}. Further details on the projections involved in this section are given in Appendix \ref{sec:proj}.

		\subsubsection{Abundance and variability estimation} \label{pp:update_A_dM}
		A direct application of \cite{Bolte2013} under the constraints \eqref{eq:constraints} leads to the following abundance update rule \vspace{-0.2cm}
\begin{align}
\label{eq:update_A}
\A_t^{(k+1)} &= \mathcal{P}_{\Ak} \left( \A_t^{(k)} - \frac{1}{L_{1t}^{(k)}} \nabla_{\A} f(\Y_t,\M^{(t)},\A_t^{(k)},\dM_t^{(k)})\right) \vspace{-0.2cm}
\end{align}
where $L_{1t}^{(k)}$ is the Lipschitz constant of $\nabla_{\A} f(\Y_t,\M^{(t)},\cdot,\dM_t^{(k)})$ and \vspace{-0.2cm}
\begin{equation}
\begin{split}
\nabla_{\A} f(&\Y_t, \M^{(t)},\A_t,\dM_t) =  \alpha (\A_t - \A_{t-1}) \\%
& + \t{(\M^{(t)} + \dM_t)} \bigl[(\M^{(t)} + \dM_t)\A_t - \Y_t \bigr] 
\end{split} \vspace{-0.3cm}
\end{equation}
\begin{equation}
L_{1t}^{(k)} = \normf{\t{(\M^{(t)} + \dM_t^{(k)})}(\M^{(t)} + \dM_t^{(k)}) + \alpha \I_\nendm}. \vspace{-0.2cm}
\end{equation}
Note that the projection $\mathcal{P}_{\Ak}$ can be exactly computed using the algorithms proposed in \cite{Duchi2008,Condat2015}. Similarly, the update rule for the variability terms is \vspace{-0.2cm}
\begin{align}
\label{eq:update_dM}
\begin{split}
&\dM_t^{(k+1)} = \\
& \mathcal{P}_{\mathcal{D}_t} \left( \dM_t^{(k)} - \frac{1}{L_{2t}^{(k)}} \nabla_{\dM} f(\Y_t,\M^{(t)},\A_t^{(k+1)},\dM_t^{(k)}) \right)
\end{split} \vspace{-0.2cm}
\end{align}
where $L_{2t}^{(k)}$ is the Lipschitz constant of $\nabla_{\dM} f(\Y_t,\M^{(t)},\A_t^{(k+1)},\cdot)$ and \vspace{-0.2cm}
\begin{equation}
\begin{split}
\nabla_{\dM} f(\Y_t,&\M^{(t)},\A_t,\dM_t) = \gamma (\dM_t - \dM_{t-1}) \\
& + \bigl[(\M^{(t)} + \dM_t)\A_t - \Y_t \bigr] \t{\A_t}
\end{split} \vspace{-0.3cm}
\end{equation}
\begin{equation}
L_{2t}^{(k)} = \normf{\A_t^{(k+1)} \A_t^{(k+1)\text{T}} + \gamma \I_\nendm}. \vspace{-0.2cm}
\end{equation}
Note that the projection $\mathcal{P}_{\D_t}$ can be efficiently approximated using the Dykstra algorithm (see \cite{Boyle1986,Becker2011,Heylen2013}). The resulting algorithm is summarized in Algo. \ref{alg:update_A_dM}.

		\subsubsection{Endmember estimation} \label{pp:update_M}

Similarly to \ref{pp:update_A_dM}, a direct application of the method detailed in \cite{Mairal2010,Bolte2013} yields \vspace{-0.2cm}
\begin{align}
\label{eq:update_M}
\M^{(t,k+1)} &= \mathcal{P}_{+}\left( \M^{(t,k)} - \frac{1}{L_{3t}} \nabla_{\M} \gt (\M^{(t,k)}) \right) \vspace{-0.2cm}
\end{align}
where $\mathcal{P}_{+}$ is the projector on $\left\{ \mathbf{X} \middle| \mathbf{X} \succeq \Zero{\nband}{\nendm} \right\}$ and $L_{3t}$ denotes the Lipschitz constant of $\nabla_{\M} \gt (\M^{(t,k)})$. Note that \vspace{-0.2cm}
\begin{align} 
\nabla_{\M} \gt(\M) &= \M \left( \frac{1}{t} \mathbf{C}_t + \beta \sum_{r=1}^\nendm \mathbf{G}_r \t{\mathbf{G}_r} \right) - \frac{1}{t} \mathbf{D}_t \\
L_{3t} &= \normf{ \frac{1}{t} \mathbf{C}_t + \beta \sum_r \mathbf{G}_r \t{\mathbf{G}_r}}. \vspace{-0.2cm}
\end{align}
The resulting algorithm is summarized in Algo. \ref{alg:update_M}.

	\subsection{Convergence guarantee} \label{p:cv_proof}
To ensure the convergence of the generated endmember sequence $(\M^{(t)})_t$ towards a critical point of the problem \eqref{eq:general_problem}, we make the following assumptions.

\begin{assumption} \label{hyp:gt}
The quadratic functions $\gt$ are strictly convex and admit a Hessian matrix lower-bounded in norm by a constant $\mu_\M >0$.
\end{assumption}
\begin{assumption} \label{hyp:Lipschitz}
The penalty functions $\Phi$, $\Psi$ and $\Upsilon$ are gradient Lipschitz continuous with Lipschitz constant $c_\Phi$, $c_\Psi$ and $c_\Upsilon$ respectively. In addition, $\Phi$ and $\Upsilon$ are assumed to be twice continuously differentiable.
\end{assumption}
\begin{assumption} \label{hyp:hessian}
The function  $f(\Y_t,\cdot,\cdot,\cdot)$ is twice continuously differentiable. The Hessian matrix of $f(\Y_t,\M,\cdot,\cdot)$ -- denoted by $\Hf$ -- is invertible at each critical point $(\A_t^*,\dM_t^*) \in \Qt$. 
\end{assumption}
In practice, Assumption \ref{hyp:gt} may be enforced by adding a penalization term $\frac{\mu_\M}{2} \normF2{\M}$ to the objective function $\gt$, where $\mu_\M$ is a small positive constant. Note that $\mu_\M$ is only a technical guarantee used in the convergence proof reported in Appendix \ref{sec:proof}, which should not be computed explicitly to be able to run the algorithm.
Assumption \ref{hyp:Lipschitz} is only included here for the sake of completeness, in case other penalizations than those given in Section \ref{sec:PLMM} are considered. Indeed, this assumption is obviously satisfied by the penalizations mentioned in this work.
Assumption \ref{hyp:hessian}, crucial to Proposition \ref{prop:convergence}, is further discussed in Appendix \ref{sec:hessian} to ease the reading of this paper.

By adapting the arguments used in \cite{Mairal2010}, the convergence property summarized in Proposition \ref{prop:convergence} can be obtained.
\begin{proposition}[Convergence of $(\M^{(t)})_t$, \cite{Mairal2010}] \label{prop:convergence}
Under the assumptions \ref{hyp:gt},\ref{hyp:Lipschitz} and \ref{hyp:hessian}, the distance between $\M^{(t)}$ and the set of critical points of the hyperspectral unmixing problem \eqref{eq:problem} converges almost surely to 0 when $t$ tends to infinity.
\end{proposition}
\begin{proof}
See Appendix \ref{sec:proof}.
\end{proof}

	\subsection{Computational complexity}

Dominated by matrix-product operations, the per image overall complexity of the proposed method is of the order \vspace{-0.2cm}
\begin{equation*}
\mathcal{O}\Bigl\{ \bigl[ \nband \nendm ( \nbpix + \Nd) + \nendm^2(\nband + \nbpix) \bigr]\Np + \Ni \nband \nendm^2 \Bigr\} \vspace{-0.2cm}
\end{equation*}
where $\Nd,\Np,\Ni$ denote the number of iterations for the Dykstra algorithm involved in the variability projection \eqref{eq:update_dM}, the PALM algorithm and the endmember update respectively. To be more explicit, the computation time for one image of size $100 \times 100$ composed of $\nband = 173$ bands is approximately \SI{6}{\second} for a \textsc{Matlab} implementation with an Intel(R) Core(TM) i5-4670 CPU @ 3.40GHz. Note that the PALM iterations (Algo. \ref{alg:update_A_dM}) and the endmember updates (Algo. \ref{alg:update_M}) can be parallelized if needed due to the separability of the objective function $f$ chosen (separability \wrt the column of the abundance matrix, and \wrt the rows of the endmember and variability matrices).

\begin{figure}[!t]
\centering
\subfloat{
\includegraphics[keepaspectratio,height=0.2\textheight , width=0.23\textwidth]{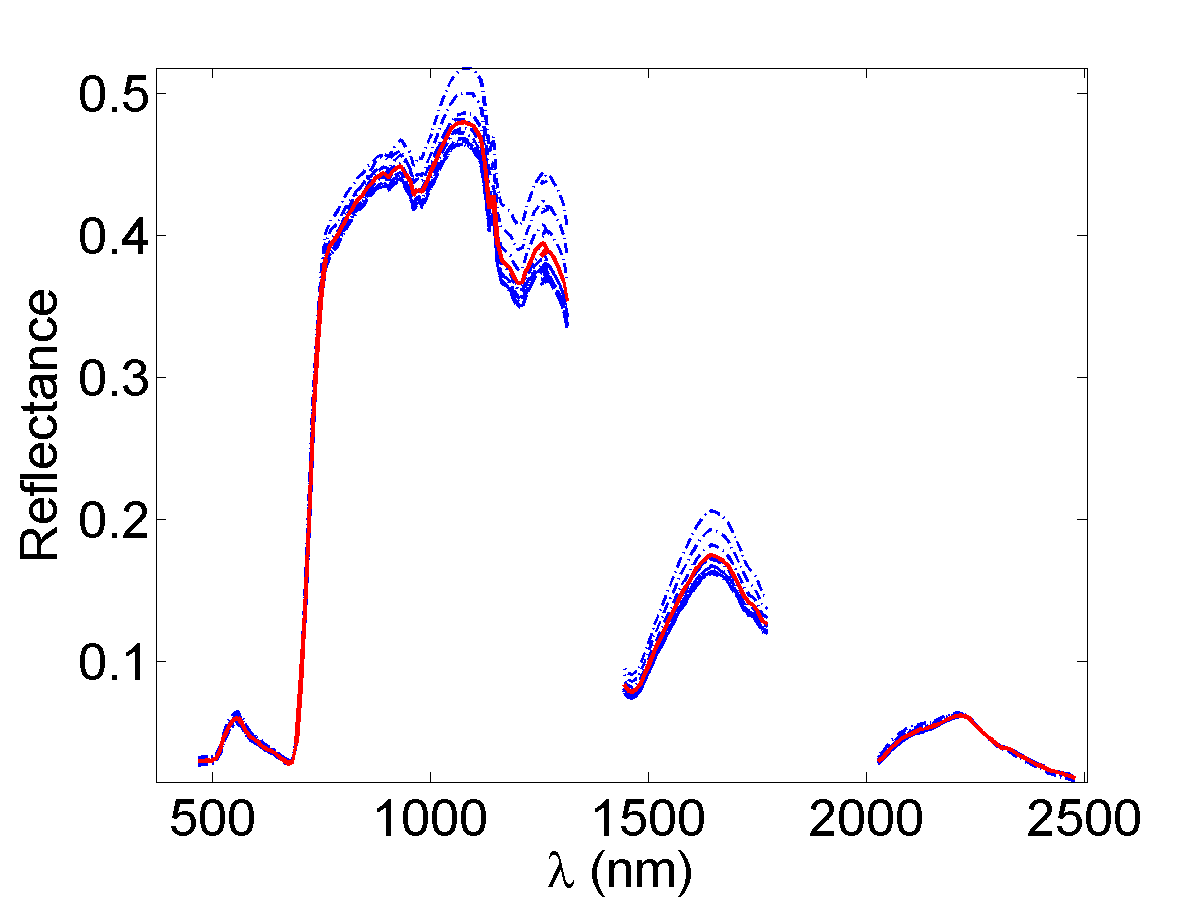}}
\subfloat{
\includegraphics[keepaspectratio,height=0.2\textheight , width=0.23\textwidth]{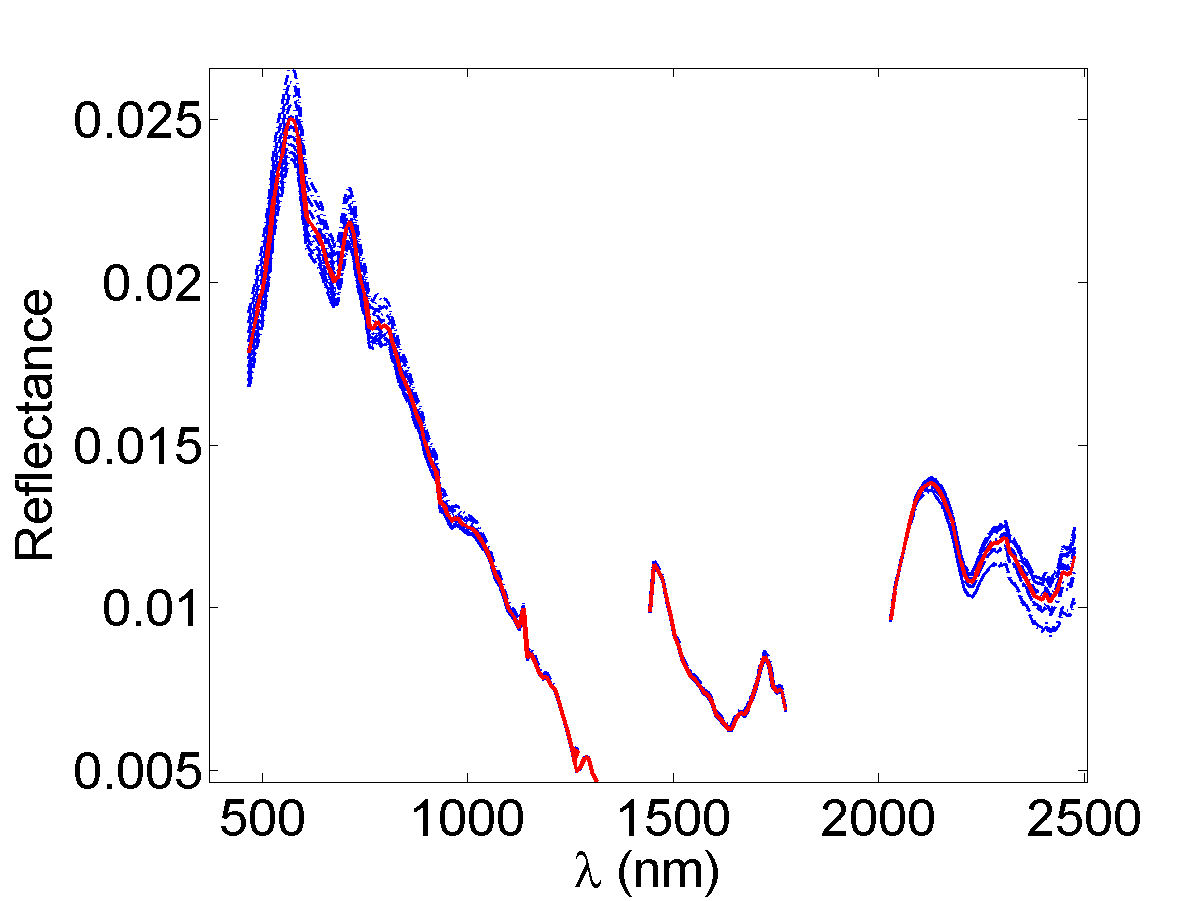}}
\vspace{-0.25cm}
\\
\subfloat{
\includegraphics[keepaspectratio,height=0.2\textheight , width=0.23\textwidth]{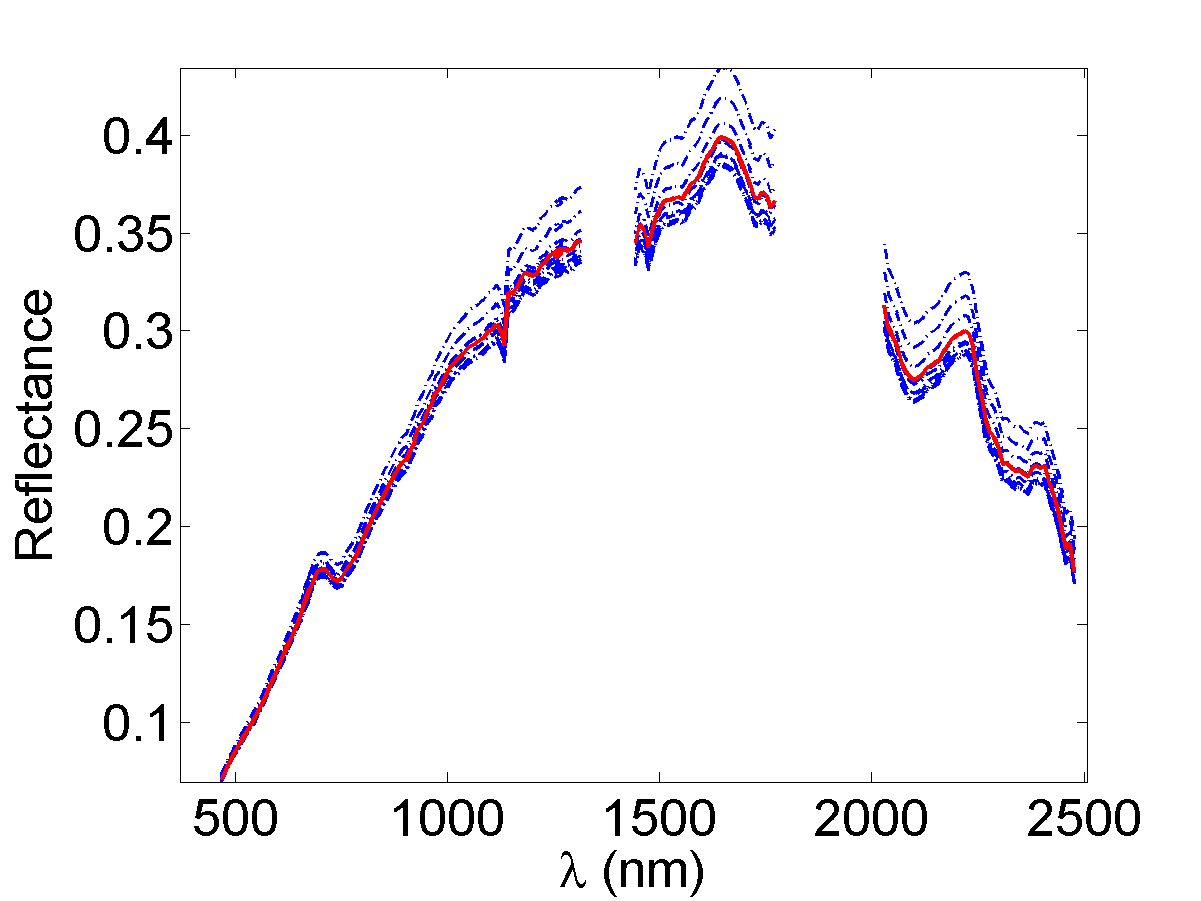}}
\caption{Reference endmembers (red lines) and the corresponding instances under spectral variability (blue lines) involved in the synthetic HS images.}
\label{fig:endm_variability} \vspace{-0.5cm}
\end{figure}


\begin{figure*}[t]
\centering
\includegraphics[keepaspectratio,height=0.3\textheight , width=\textwidth]{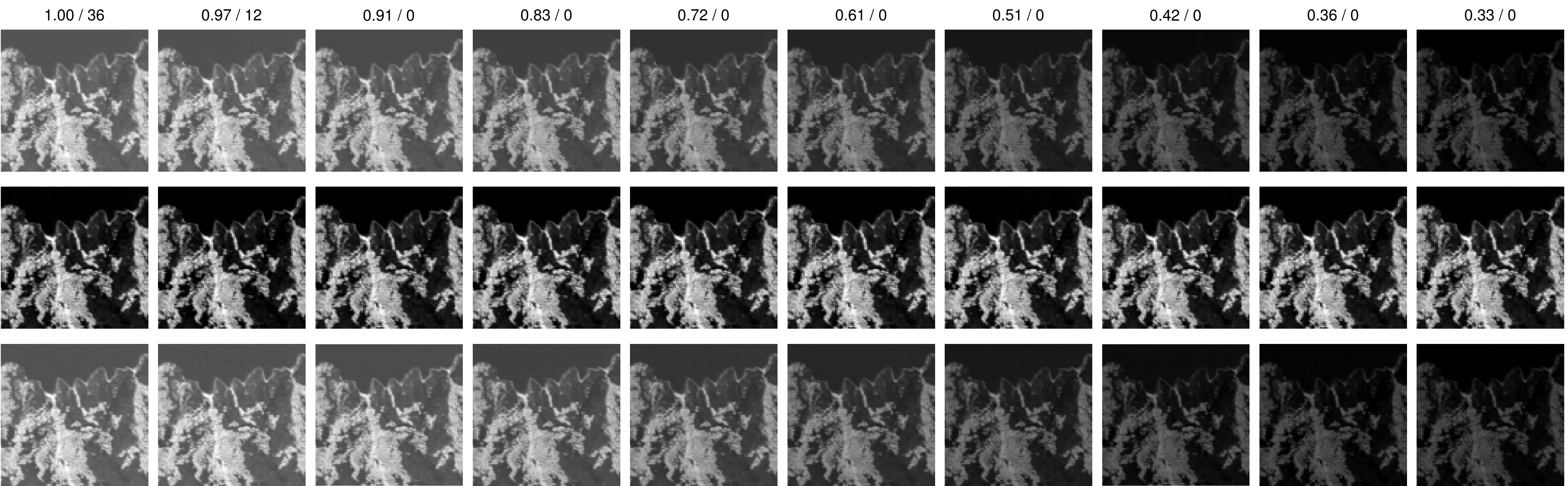}
\caption{Abundance maps of the first endmember used in the synthetic mixtures (theoretical abundances on the first line, VCA/FCLS on the second line, proposed method on the third line). The top line indicates the theoretical maximum abundance value and the true number of pixels whose abundance is greater than 0.95 for each time instant.}
\label{fig:synth_abundance1} \vspace{-0.5cm}
\end{figure*}
\begin{figure*}[th]
\centering
\includegraphics[keepaspectratio,height=0.3\textheight , width=\textwidth]{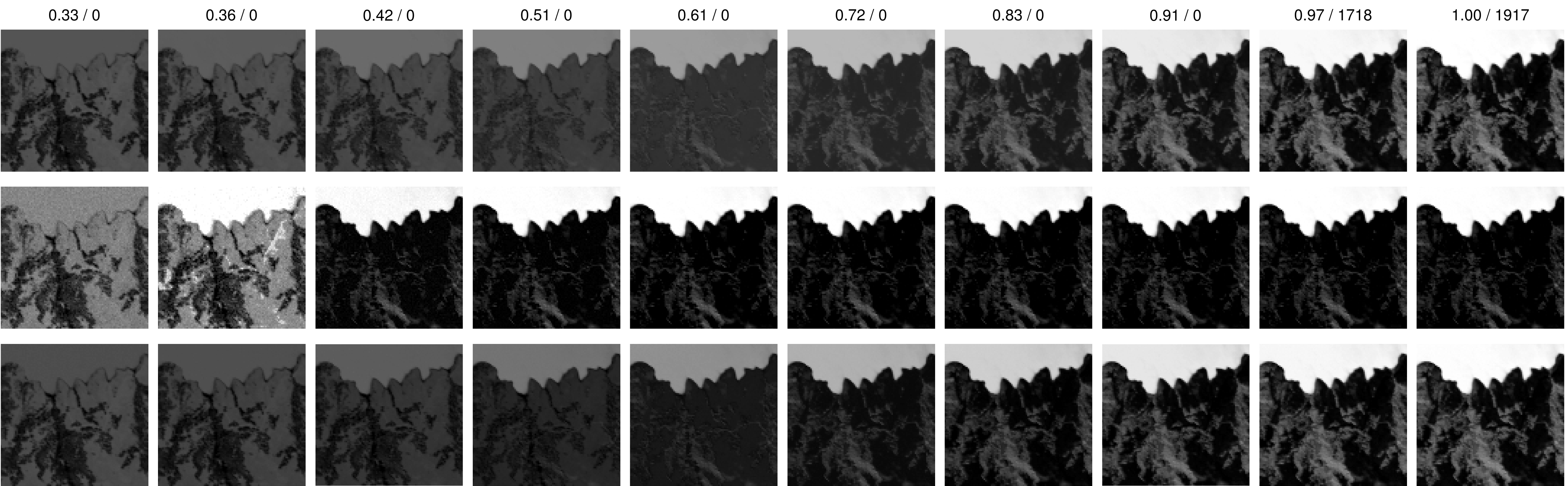}
\caption{Abundance maps of the second endmember used in the synthetic mixtures (theoretical abundances on the first line, VCA/FCLS on the second line, proposed method on the third line). The top line indicates the theoretical maximum abundance value and the true number of pixels whose abundance is greater than 0.95 for each time instant.}
\label{fig:synth_abundance2} \vspace{-0.5cm}
\end{figure*}
\begin{figure*}[t]
\centering
\includegraphics[keepaspectratio,height=0.3\textheight , width=\textwidth]{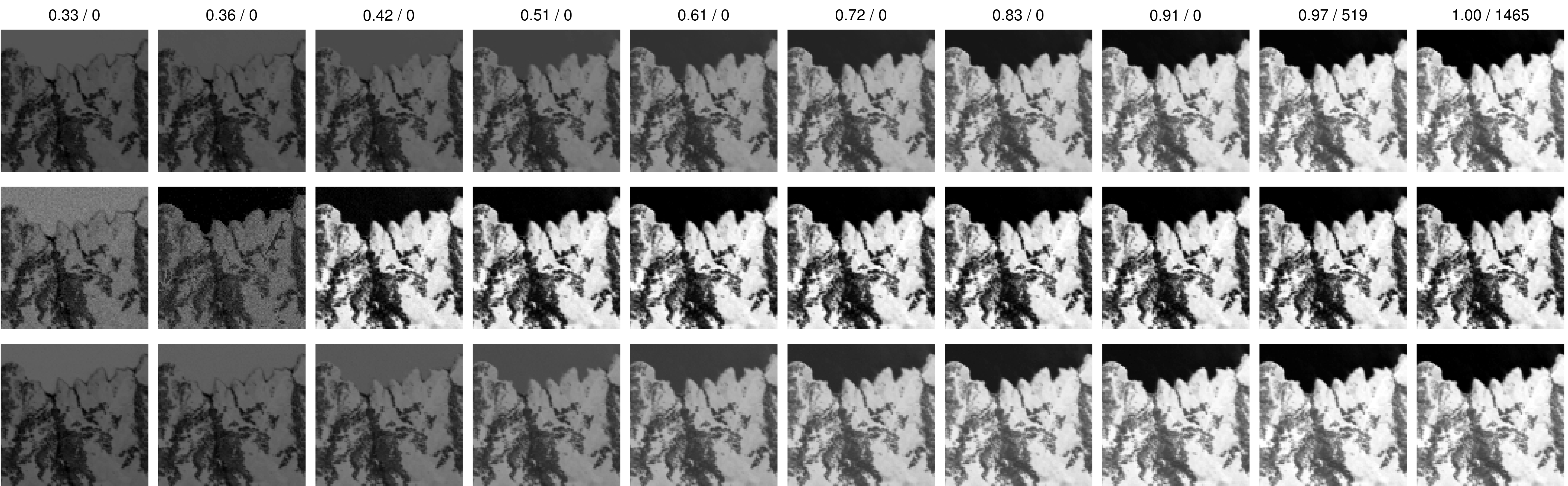}
\caption{Abundance maps of the third endmember used in the synthetic mixtures (theoretical abundances on the first line, VCA/FCLS on the second line, proposed method on the third line). The top line indicates the theoretical maximum abundance value and the true number of pixels whose abundance is greater than 0.95 for each time instant.}
\label{fig:synth_abundance3} \vspace{-0.7cm}
\end{figure*}

\begin{figure}[tbhp]
\centering
\subfloat[1][]{
\includegraphics[keepaspectratio,height=0.15\textheight , width=0.15\textwidth]{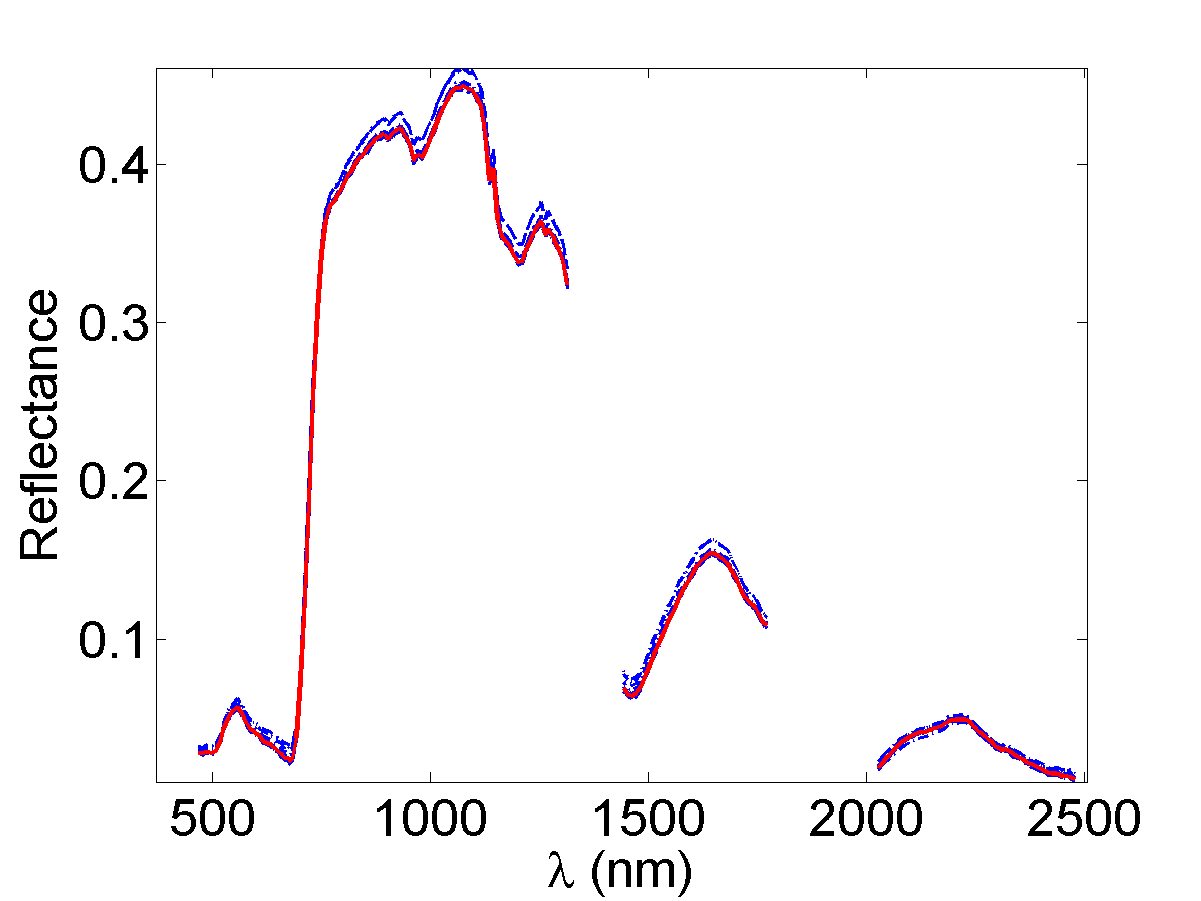}
\label{fig:synth_endm1}}
\subfloat[2][]{
\includegraphics[keepaspectratio,height=0.15\textheight , width=0.15\textwidth]{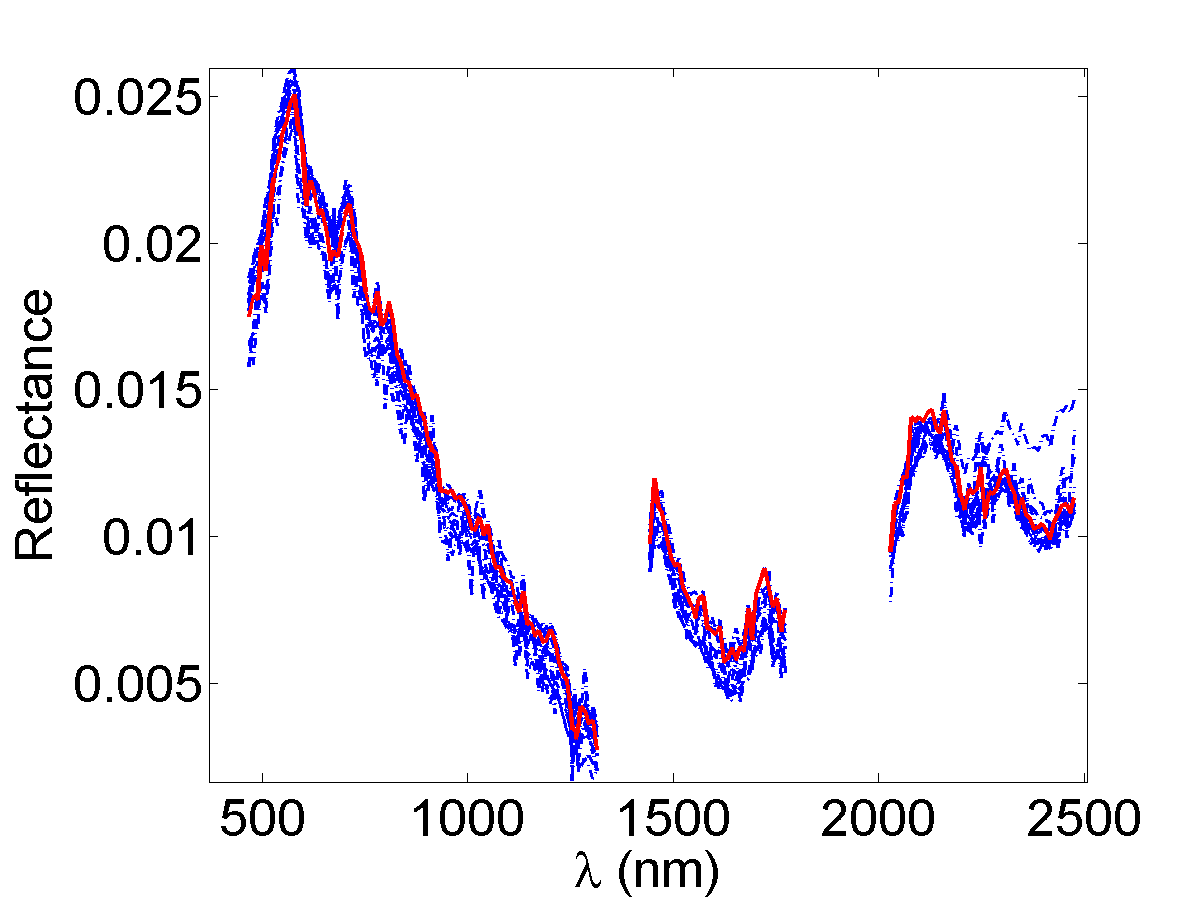}
\label{fig:synth_endm2}}
\subfloat[3][]{
\includegraphics[keepaspectratio,height=0.15\textheight , width=0.15\textwidth]{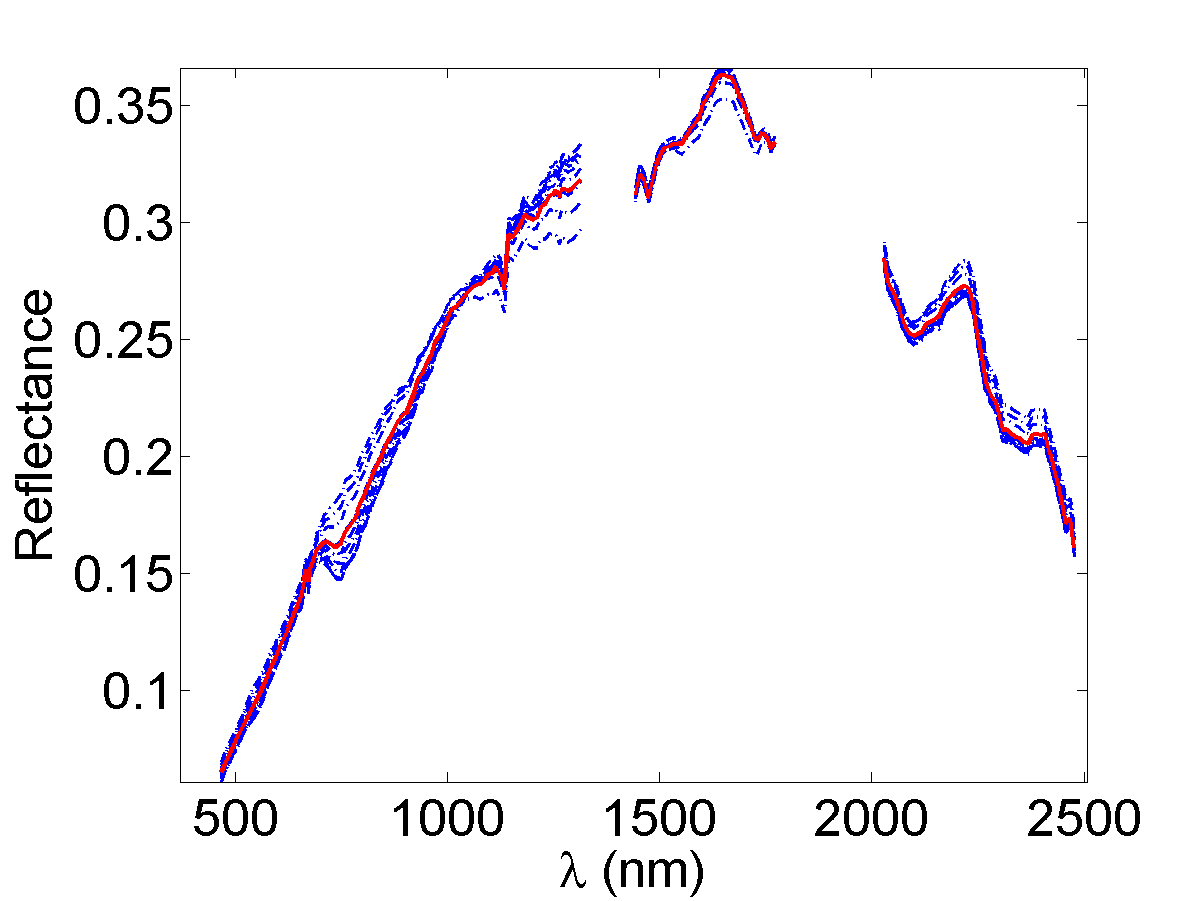}
\label{fig:synth_endm3}}
\vspace{-0.25cm}
\\
\subfloat[1][(VCA)]{
\includegraphics[keepaspectratio,height=0.15\textheight , width=0.15\textwidth]{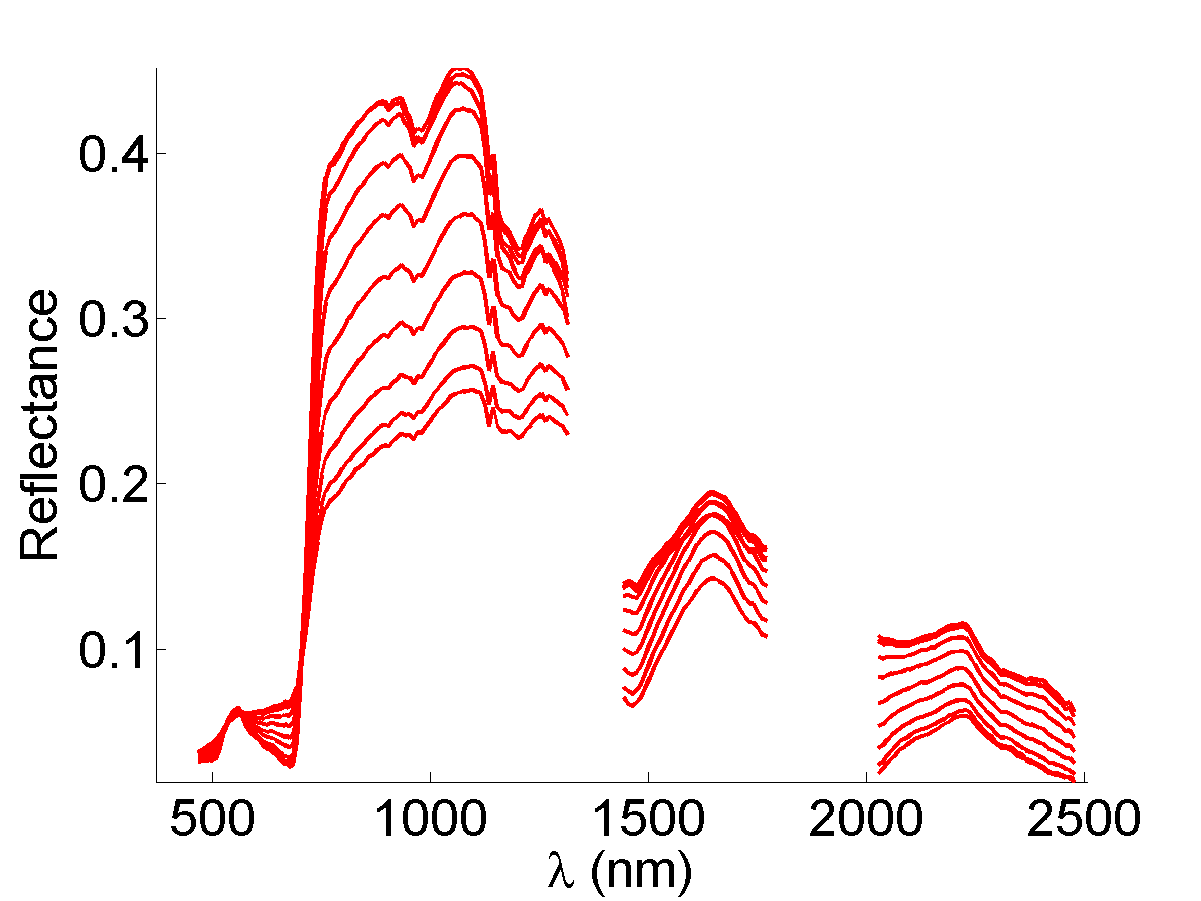}
\label{fig:synth_vca_endm1}}
\subfloat[2][(VCA)]{
\includegraphics[keepaspectratio,height=0.15\textheight , width=0.15\textwidth]{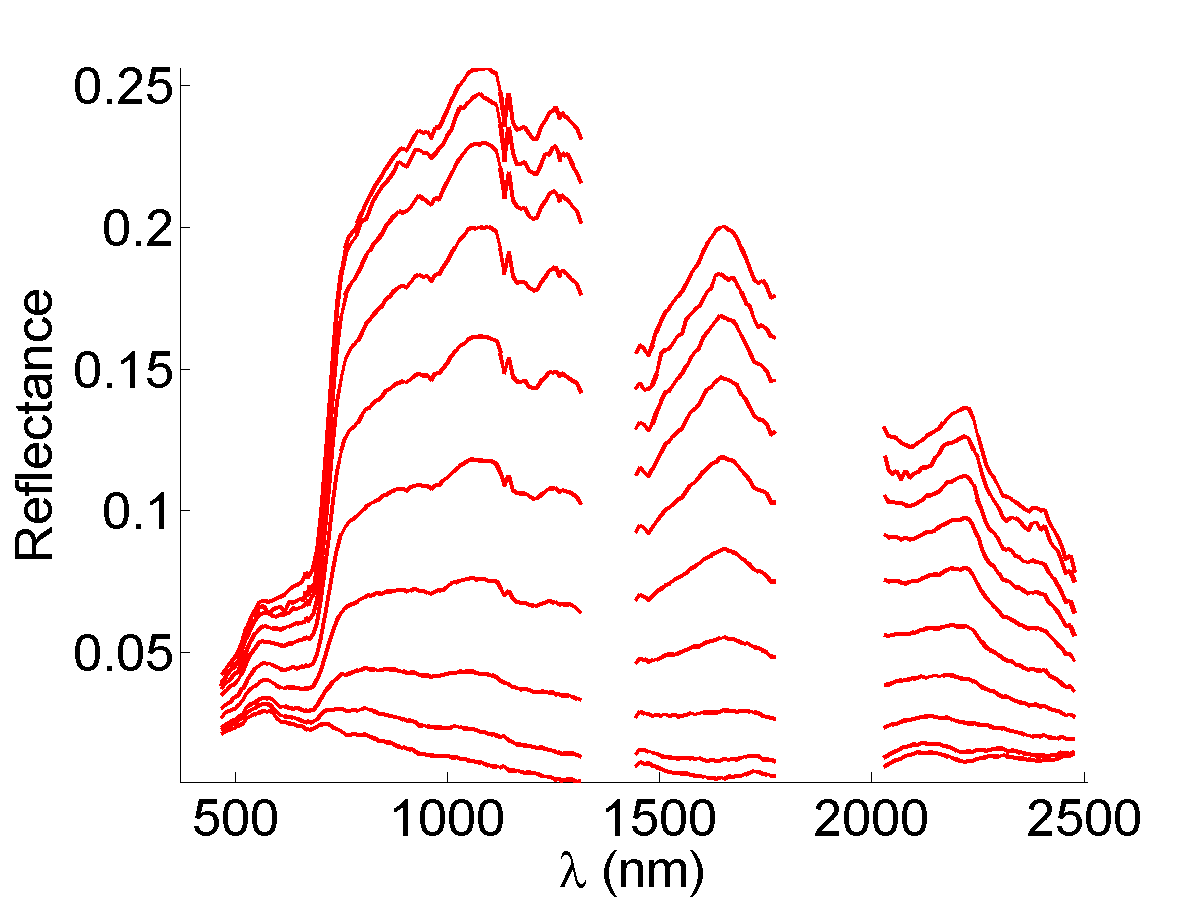}
\label{fig:synth_vca_endm2}}
\subfloat[3][(VCA)]{
\includegraphics[keepaspectratio,height=0.15\textheight , width=0.15\textwidth]{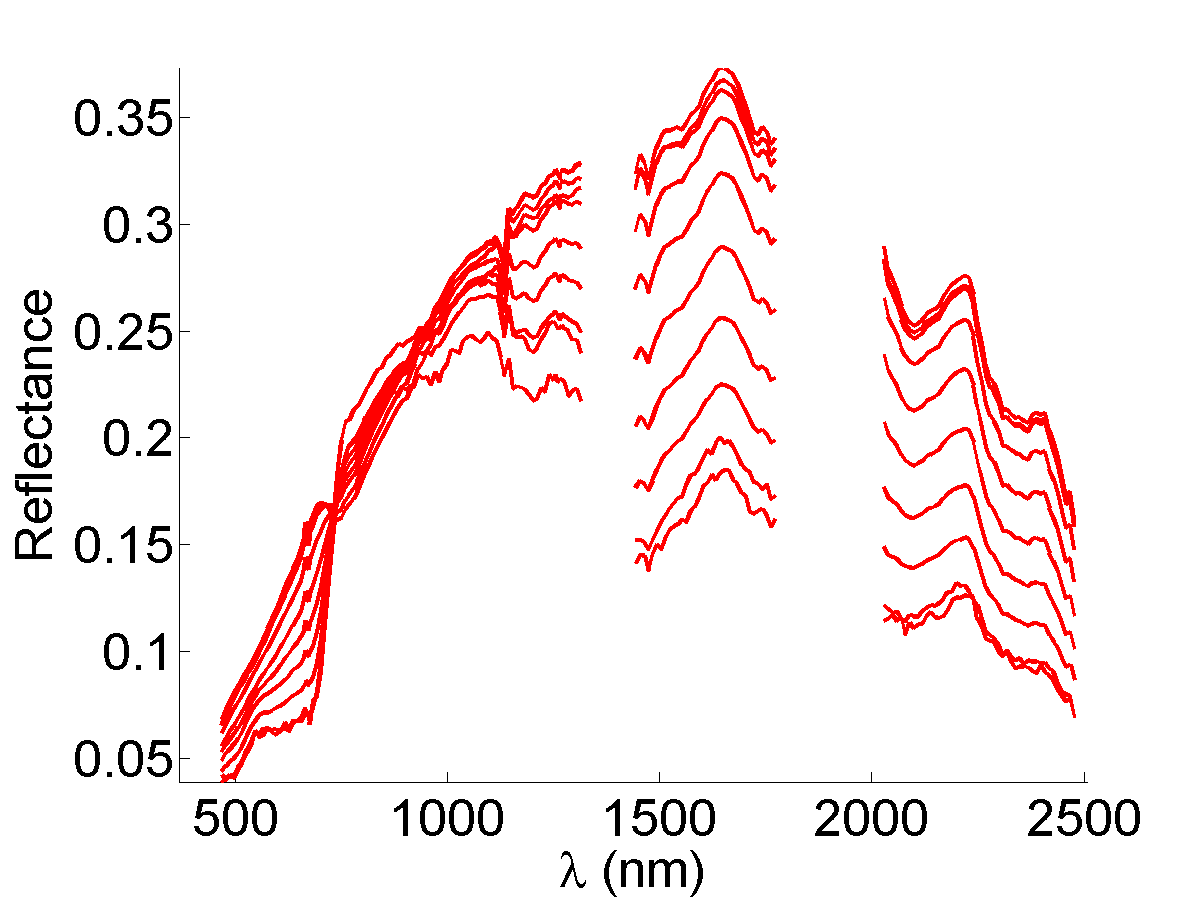}
\label{fig:synth_vca_endm3}}
\caption{Estimated endmembers on the synthetic hyperspectral time series (PLMM endmembers in red with variability in blue dotted lines on the first line, VCA-extracted endmembers on the second line, SISAL endmembers are omitted since very similar to those obtained with VCA).}
\label{fig:synth_endm} \vspace{-0.5cm}
\end{figure}
\begin{table}[!t]
\caption{Parameters used in the experiments.}
\label{tab:param}
	\begin{center}
		\begin{tabular}{@{}lll@{}} \toprule
 		 					& Synthetic data	& Real data     \\ \cmidrule{2-3} 	
$\sigma^2$					& 1	   & 1		 	\\ 
$\kappa^2$					& 0.1  & 0.01		 \\		 
$\alpha$					& $10^{-4}$   & 0		 	\\ 
$\beta$						& $10^{-3}$	& $10^{-4}$		 	\\
$\gamma$					& $3 \times 10^{-5}$	& 0		 	\\ 
$\Nd$			 			& 50	   & 50	  	 	\smallskip\\
$\Np$ 						& 50	   & 50	  	 	\smallskip\\
$\Ni$ 						& 50	   & 50	  	 	\smallskip\\
$\Ne$ 						& 10	   & 10	  	 	\\
$\xi$			 			& 0.98	   & 0.98	  	 	\\ \bottomrule
		\end{tabular}
	\end{center} \vspace{-0.5cm}
\end{table} 

\section{Experiment with synthetic data} \label{sec:simulations}
	This section considers an HS image sequence composed of 10 images of size $98~\times~102$, each image composed of $173$ bands. The images correspond to linear mixtures of 3, 6 and 10 endmembers affected by smooth time-varying variability. The synthetic abundance maps of this scenario vary smoothly from one image to another. Note that the pure pixel assumption is not satisfied for all images of the experiment with $\nendm~=~3$ endmembers in order to assess the algorithm performance in a challenging scenario. The synthetic linear mixtures have been corrupted by additive white Gaussian noise to ensure a resulting signal-to-noise ratio of $\text{SNR} = \SI{30}{\text{dB}}$. Additional results for mixtures corrupted by colored Gaussian noise are available in \cite[App. D]{Thouvenin2015TR}.

	In order to introduce controlled spectral variability, the endmembers involved in the mixtures have been generated using the product of reference endmembers with randomly generated piecewise-affine functions as in \cite{Thouvenin2015}. The corresponding perturbed endmembers used in the experiment are depicted in Fig. \ref{fig:endm_variability}. Note that different affine functions have been considered at each time instant for each endmember.
			
	\subsection{Compared methods}
The results of the proposed algorithm have been compared to those obtained with several classical linear unmixing methods performed individually on each image of the time series. The methods are recalled below with their most relevant implementation details. All the methods requiring an appropriate initialization have been initialized with VCA/FCLS.
\begin{enumerate}
\item VCA/FCLS (no variability): for each image, the endmembers are first extracted using the vertex component analysis (VCA) \cite{Nascimento2005} which requires pure pixels to be present in the analyzed images. The abundances are then estimated for each pixel by solving a Fully Constrained Problem (FCLS) with ADMM \cite{Bioucas2010}; 
\item SISAL/FCLS (no variability): the endmembers are first extracted using the simplex identification via split augmented Lagrangian (SISAL) \cite{Bioucas2009}. Note that the pure pixel assumption is not required to apply this method. The tolerance for the stopping rule has been set to $10^{-3}$. The abundances are then estimated for each pixel by FCLS;
\item $\ell_{1/2}$ NMF (no variability): the algorithm described in \cite{Qian2011} is applied to each image, with a stopping criterion set to $10^{-3}$ and a maximum of $300$ iterations. The regularization parameter has been set as in \cite{Qian2011};
\item BCD/ADMM: the algorithm described in \cite{Thouvenin2015} is applied to each image with a stopping criterion set to $10^{-3}$. The endmember regularization recalled in \eqref{eq:reg_M} has been used, with a parameter set to the same value as the one used for the proposed method. The abundance regularization parameter (spatial smoothness) has been set to $10^{-4}$, and the variability regularization parameter has been set to 1;
\item Proposed method: endmembers are initialized with VCA applied to the union of the pixels belonging to the $\nendm-1$ convex hull of each image. The abundances are initialized by FCLS, and the variability matrices are initialized with all their entries equal to $0$. Whenever the algorithm is applied to a previously processed image, the previous abundance and variability estimates are taken as a warm-restart. Algo. \ref{alg:update_A_dM} (PALM algorithm) is stopped after $\Np$ iterations and the Dykstra algorithm used to compute the projection in \eqref{eq:update_dM} is iterated $\Nd$ times. Moreover, Algo. \ref{alg:update_M} is stopped after $\Ni$ iterations. Finally, Algo. \ref{alg:SAA} is stopped after $\Ne$ cycles -- referred to as epochs -- on the randomly permuted training set to approximately obtain i.i.d. samples \cite{Mairal2010}. In particular, the number of cycles $N_{\text{epochs}}$ and sub-iterations $N_{\text{iter}}$ have been empirically chosen to obtain a compromise between the estimation accuracy and the implied computational cost. We also included a constant forgetting factor $\xi \in (0,1)$ in order to slowly forget the past data. The closer to one $\xi$ is, the more slowly the past data are forgotten.
\end{enumerate}

The performance of the algorithm has been assessed in terms of endmember estimation using the average spectral angle mapper (aSAM) defined as \vspace{-0.2cm}
\begin{equation}
\aSAM(\M) = \frac{1}{\nendm } \sum_{r=1}^\nendm   \arccos \left( \frac{ \t{\mathbf{m}_r} \widehat{\mathbf{m}}_r } { \lVert \mathbf{m}_r \rVert_2 \lVert \widehat{\mathbf{m}}_r \rVert_2   } \right) \vspace{-0.2cm}
\end{equation}{}%
as well as in terms of abundance and perturbation estimation through the global mean square errors (GMSEs) \vspace{-0.2cm}
\begin{align}
\GMSE(\A)  & = \frac{1}{\ntime \nendm \nbpix} \sum_{t=1}^\ntime \lVert \A_t - \widehat{\A}_t
\rVert_{\text{F}}^2 \\
\GMSE(\dM) & = \frac{1}{\ntime \nband \nendm } \sum_{t=1}^\ntime \lVert \dM_t - \widehat{\dM}_t \rVert_{\text{F}}^2 . \vspace{-0.2cm}
\end{align}
As a measure of fit, the following reconstruction error (RE) has been considered \vspace{-0.2cm}
\begin{align}
\label{eq:RE}
\RE &= \frac{1}{\ntime \nband \nbpix} \sum_{t=1}^\ntime \lVert \Y_t - \widehat{\Y}_t \rVert_{\text{F}}^2
%
\end{align}
where $\widehat{\Y}_t$ is the matrix formed of the pixels reconstructed with the parameters estimated for the image $t$. 


	\subsection{Results}

The parameters used for the proposed algorithm, which have been adjusted by cross-validation, are detailed in Table \ref{tab:param}. For the dataset associated with mixtures of $\nendm = 3$ endmembers, the abundance maps obtained by the proposed method are compared to those of VCA/FCLS in Figs. \ref{fig:synth_abundance1} to \ref{fig:synth_abundance3}, whereas the corresponding endmembers are displayed in Fig. \ref{fig:synth_endm}. The abundance maps obtained by SISAL/FCLS, $\ell_{1/2}$ NMF and BCD/ADMM, somewhat similar to those obtained by VCA/FCLS, are included in a separate technical report \cite{Thouvenin2015TR}, along with a more detailed version of Table \ref{tab:results_synth} and the endmembers extracted by all the unmixing strategies. The performance of the unmixing methods is finally reported in Table \ref{tab:results_synth}, leading to the following conclusions.

\begin{itemize}
\item The proposed method is more robust to the absence of pure pixels in some images than both VCA/FCLS and SISAL/FCLS. Note that $\ell_{1/2}$ NMF and BCD/ADMM converge to poor local optima, which directly results from the poor performance of VCA in this specific context. On the contrary, the estimated abundances obtained with the proposed method (second line of Figs. \ref{fig:synth_abundance1} to \ref{fig:synth_abundance3}) are closer to the ground truth (first line) than VCA/FCLS (third line). This observation is confirmed by the results given in Table \ref{tab:results_synth};
\item The proposed method provides competitive unmixing results while allowing temporal endmember variability to be estimated for each endmember (see Fig. \ref{fig:synth_endm});
\item The abundance GMSEs and the REs estimated with the proposed method are lower or comparable to those obtained with VCA/FCLS and SISAL/FCLS applied to each image individually (see Table \ref{tab:results_synth}), without introducing much more degrees of freedom into the underlying model when compared to BCD/ADMM;
\item Even though the performance of the proposed method degrades with the number of endmembers, the results remain better or comparable to those of the other methods.
\end{itemize}

Whenever an endmember is scarcely present in one of the images, the proposed method outperforms VCA/FCLS as can be seen in Figs. \ref{fig:synth_abundance1} to \ref{fig:synth_abundance3}. Note that the maximum theoretical abundance value and the number of pixels whose abundances are greater than 0.95 are mentioned on the top line of Figs. \ref{fig:synth_abundance1} to \ref{fig:synth_abundance3}, to assess the difficulty of recovering each endmember in each image. This result was expected, since VCA is a pure pixel-based unmixing method. 

\setlength\columnsep{0.1pt}
\begin{table}[!t] 
\vspace{-0.5cm}
\centering
\caption{Simulation results on synthetic data (aSAM($\M$) in (\textdegree), GMSE($\mathbf{A}$)$\times 10^{-2}$, GMSE($\mathbf{dM}$)$\times 10^{-4}$, RE $\times 10^{-4}$, time in (s)).}
	\begin{center}
	\resizebox{0.48\textwidth}{!}{%
		\begin{tabular}{@{}llccccc@{}} \toprule
&		   	  & aSAM($\M$) & GMSE($\A$) & GMSE($\dM$) & RE & time \\ \cmidrule{1-7}
\multirow{5}{*}{\rotatebox{90}{$\nendm = 3$}}
&VCA/FCLS      &  15.76  & 4.22 &/ & 0.413 & \textbf{1} \\
&SISAL/FCLS	  &  15.88  & 3.68 &/ & \textbf{0.375} & 5	\\
&$\ell_{1/2}$ NMF  & 18.54 & 8.20 &/ & 0.123 & 226 \\
&BCD/ADMM	      & 15.91 & 3.47 & 4.03 & 0.282 & 379	\\
&Proposed 	  & \textbf{1.88} & \textbf{0.23} & \textbf{1.02} & \textbf{0.375} & 168 \\ \cmidrule{1-7}
\multirow{5}{*}{\rotatebox{90}{$\nendm = 6$}}
&VCA/FCLS      &  2.14  & \textbf{0.14} & / & 1.48 & \textbf{4} \\
&SISAL/FCLS	  &  1.67  & 0.83 & / & \textbf{1.20} & 5	\\
&$\ell_{1/2}$ NMF  & 3.41 & 0.45 & / & 1.53 & 332 \\
&BCD/ADMM	      & 2.27 & 0.29 & \textbf{1.31} & \textbf{1.20} & 1066	\\
&Proposed 	  & \textbf{1.49} & 0.17 & 2.69 & 1.22 &344 \\ \cmidrule{1-7}
\multirow{5}{*}{\rotatebox{90}{$\nendm = 10$}}
&VCA/FCLS      &  3.52  & 7.24 & / & 4.63 & \textbf{5} \\
&SISAL/FCLS	  &  9.53  & 3.32 & / & \textbf{1.67} & 6	\\
&$\ell_{1/2}$ NMF  & 5.58 & 7.03 & / & 3.90 & 279 \\
&BCD/ADMM	      & 3.27 & 6.45 & \textbf{7.2} & 1.70 & 735 \\
&Proposed 	  & \textbf{2.83} & \textbf{0.43} & 8.9 & 1.99 & 204 \\
    \bottomrule
		\end{tabular}
}
	\end{center}
\label{tab:results_synth} \vspace{-0.2cm}
\end{table}

	\subsection{Hyper-parameter influence on the reconstruction error} \label{sec:tuning}
	
Considering the significant number of hyper-parameters to be tuned (i.e., $\alpha,\beta,\gamma,\sigma,\kappa$), a full sensitivity analysis is a challenging task, which is further complexified by the non-convex nature of the problem considered. To alleviate this issue, each parameter has been individually adjusted while the others were set to a priori reasonable values (i.e., $(\alpha,\beta,\gamma,\sigma^2,\kappa^2) = (10^{-2},10^{-4},10^{-4},\hat{\sigma}^2,10^{-3})$, where $\hat{\sigma}^2 = 0.0372$ denotes the theoretical average energy of the variability introduced in the synthetic dataset used for this analysis). The appropriateness of a given range of values has been evaluated in terms of the RE of the recovered solution. The results reported in Fig. \ref{fig:re_var} suggest that the proposed method is relatively robust to the choice of the hyper-parameters. 
More precisely, as can be seen in Figs. \ref{fig:re_M} and \ref{fig:re_dM}, only $\beta$ and $\gamma$ may induce oscillations (of very small amplitude) in the RE. Based on this analysis, it is interesting to note that the interval $[2\times 10^{-3},10^{-2}]$ can be chosen in practice to obtain reasonable reconstruction errors.

To conclude, the two following remarks can be made on the choice of $\sigma$ and $\kappa$:
\begin{itemize}
\item the value chosen for $\sigma$ results from an empirical compromise between the risk to capture noise into the variability terms ($\sigma$ too large) and the risk to lose information ($\sigma$ too small). The sensitivity analysis conducted in Fig. \ref{fig:re_sigma2} shows that $\sigma^2 \in [10^{-1},1]$ provides interesting results for this experiment;
\item $\kappa$ should be set to a value ensuring that $\M$ reflects the average spectral behavior of the perturbed endmembers. Fig. \ref{fig:re_kappa2} shows that $\kappa^2 \in [10^{-3},1]$ provides interesting results for the synthetic dataset used in the experiment.
\end{itemize}
	
\begin{figure}[t!]
\centering
\subfloat[1][]{
\includegraphics[keepaspectratio,height=0.2\textheight , width=0.2\textwidth]{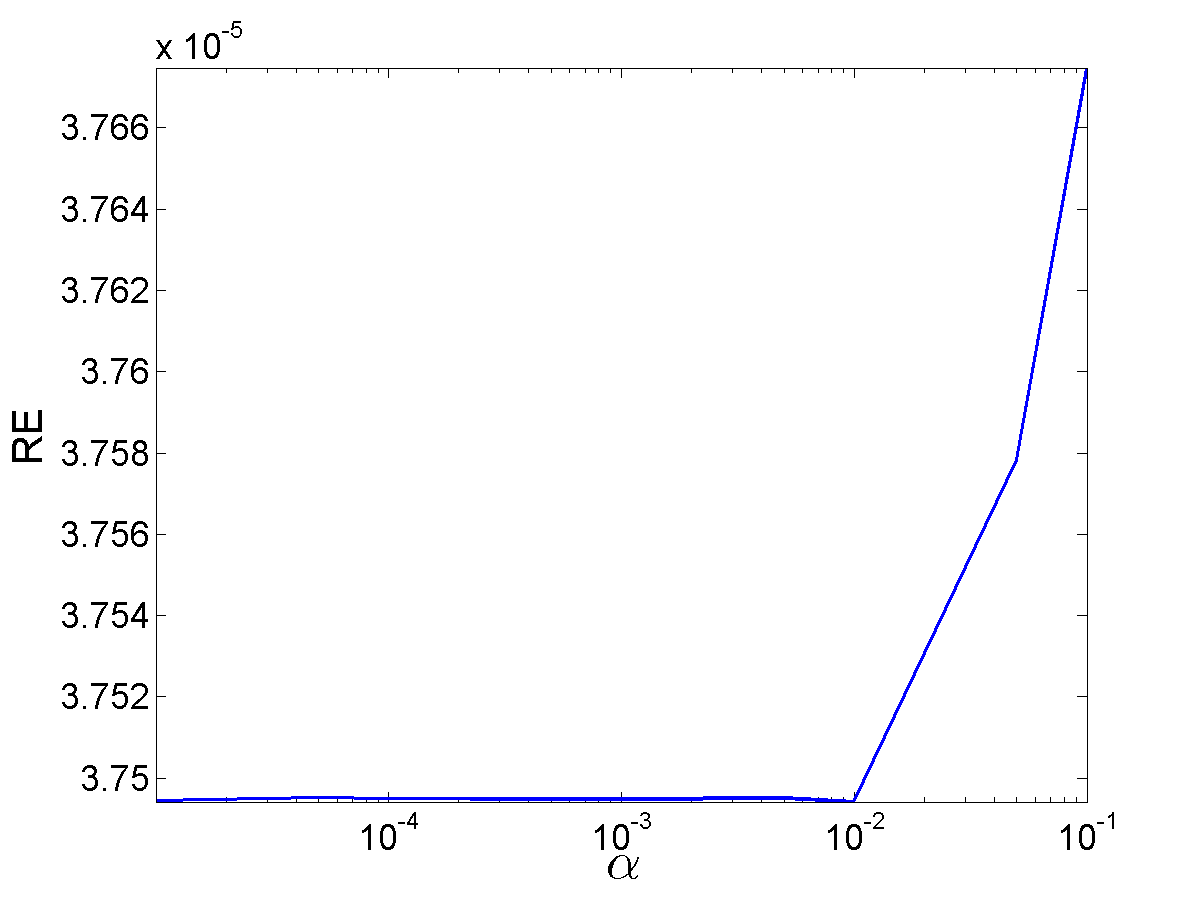}
\label{fig:re_A}}
\subfloat[2][]{
\includegraphics[keepaspectratio,height=0.2\textheight , width=0.2\textwidth]{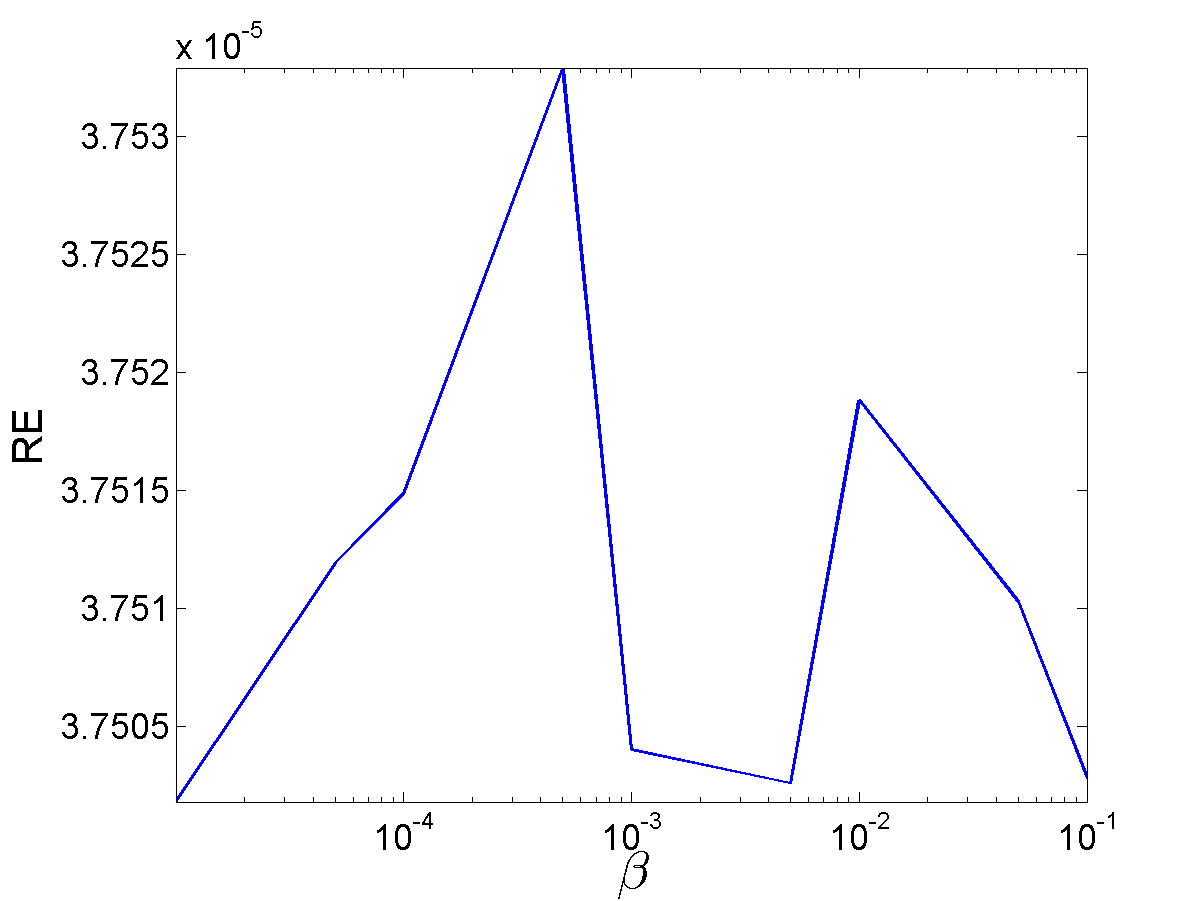}
\label{fig:re_M}}
\vspace{-0.3cm}
\\
\subfloat[3][]{
\includegraphics[keepaspectratio,height=0.2\textheight , width=0.2\textwidth]{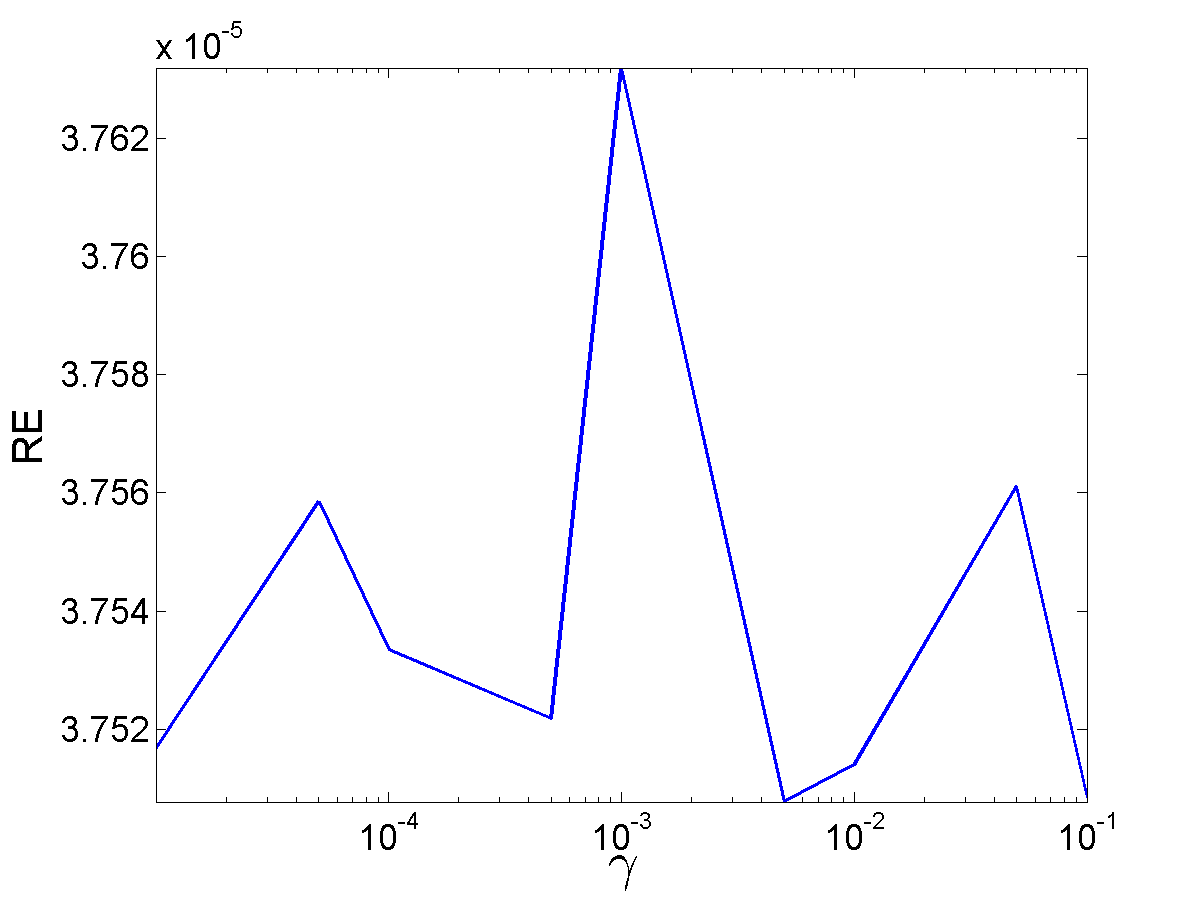}
\label{fig:re_dM}}
\subfloat[4][]{
\includegraphics[keepaspectratio,height=0.2\textheight , width=0.2\textwidth]{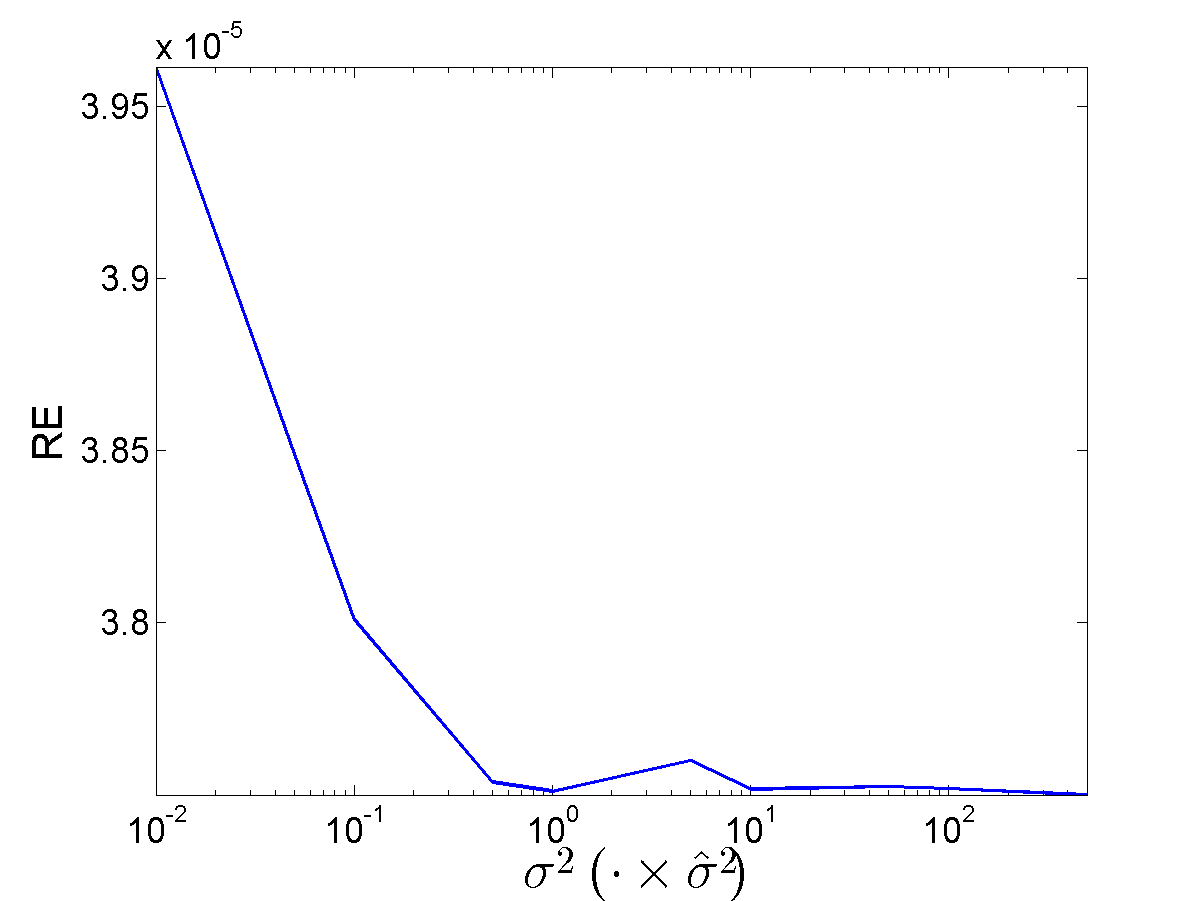}
\label{fig:re_sigma2}}
\vspace{-0.3cm}
\\
\subfloat[5][]{
\includegraphics[keepaspectratio,height=0.2\textheight , width=0.2\textwidth]{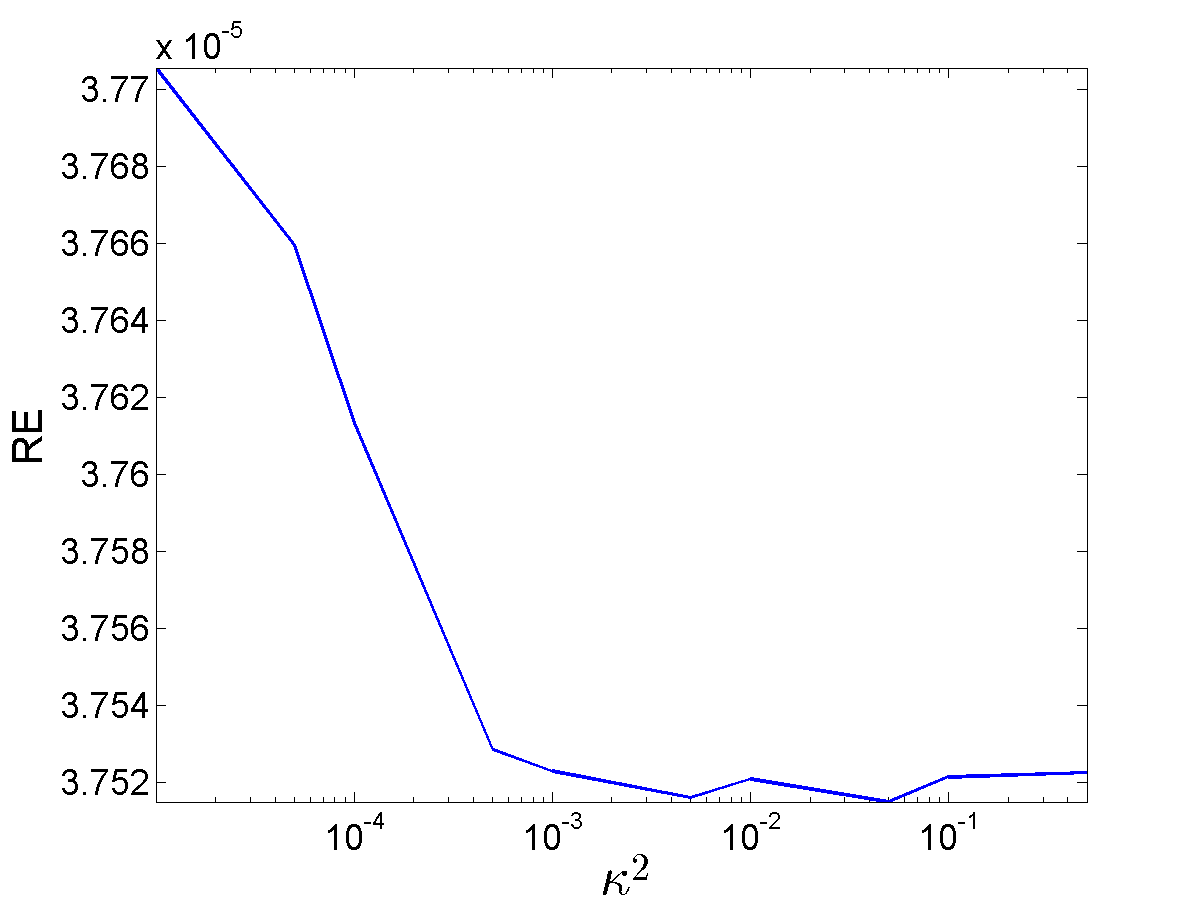}
\label{fig:re_kappa2}}
\vspace{-0.25cm}
\caption{Sensitivity analysis of the reconstruction error RE \wrt the tuning of the algorithm hyper-parameters ($\hat{\sigma}^2 = 0.0372$ denotes the theoretical average energy of the variability introduced in the synthetic dataset used for this analysis).}
\label{fig:re_var} \vspace{-0.5cm}
\end{figure}

\section{Experiment with real data} \label{sec:experiments}

	\subsection{Description of the dataset}

The proposed algorithm has been applied to real HS images acquired by the Airborne Visible Infrared Imaging Spectrometer (AVIRIS) over the Lake Tahoe region (California, United States of America) between 2014 and 2015\footnote{The images used in this experiment are freely available from the online AVIRIS flight locator tool at \url{http://aviris.jpl.nasa.gov/alt_locator/}.}. Water absorption bands were removed from the 224 spectral bands, leading to 173 exploitable bands. In absence of any ground truth, the sub-scene of interest ($150 \times 110$), partly composed of a lake and a nearby field, has been unmixed with $\nendm = 3$, 4 and 5 endmembers to obtain a compromise between the results of HySime \cite{Bioucas2008}, those of the recently proposed eigen-gap approach (EGA) \cite{Halimi2016} (see Table \ref{tab:results_r}), and the consistency of the resulting abundance maps. The parameters used for the proposed approach are given in Table \ref{tab:param}, and the other methods have been run with the same parameters as in Section \ref{sec:experiments}. Note that a $4 \times 4$ patch composed of outliers has been manually removed from the last image of the sequence prior to the unmixing procedure. 
%

\begin{figure*}[tbhp!]
\centering
\subfloat[1][04/10/2014]{
\includegraphics[keepaspectratio,height=0.15\textheight , width=0.15\textwidth]{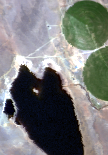}
\label{fig:cube1}}
\quad
\subfloat[2][06/02/2014]{
\includegraphics[keepaspectratio,height=0.15\textheight , width=0.15\textwidth]{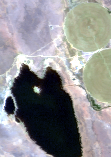}
\label{fig:cube2}}
\quad
\subfloat[3][09/19/2014]{
\includegraphics[keepaspectratio,height=0.15\textheight , width=0.15\textwidth]{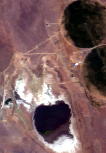}
\label{fig:cube3}}
\quad
\subfloat[4][11/17/2014]{
\includegraphics[keepaspectratio,height=0.15\textheight , width=0.15\textwidth]{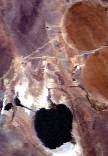}
\label{fig:cube4}}
\quad
\subfloat[5][04/29/2015]{
\includegraphics[keepaspectratio,height=0.15\textheight , width=0.15\textwidth]{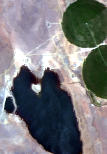}
\label{fig:cube5}}
\vspace{-0.1cm}
\caption{Scenes used in the experiment, given with their respective acquisition date.}
\label{fig:cube} \vspace{-0.7cm}
\end{figure*}

\begin{figure}[t]
\centering
\includegraphics[keepaspectratio,height=0.2\textheight , width=0.5\textwidth]{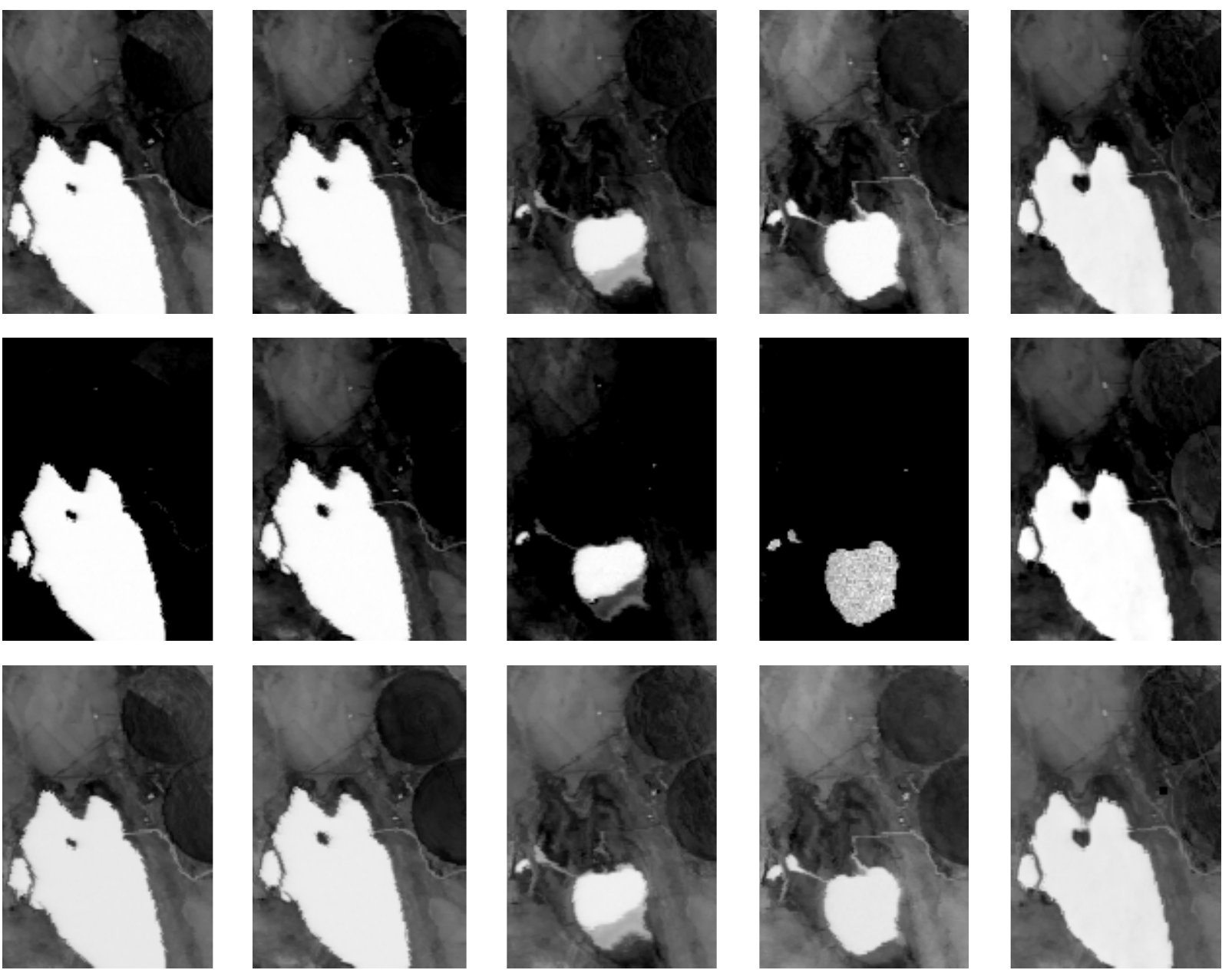}
\caption{Water abundance maps (proposed method on the first line, VCA/FCLS on the second line, SISAL/FCLS on the third line).}
\label{fig:abundance1} \vspace{-0.5cm}
\end{figure}

\begin{figure}[th]
\centering
\includegraphics[keepaspectratio,height=0.2\textheight , width=0.5\textwidth]{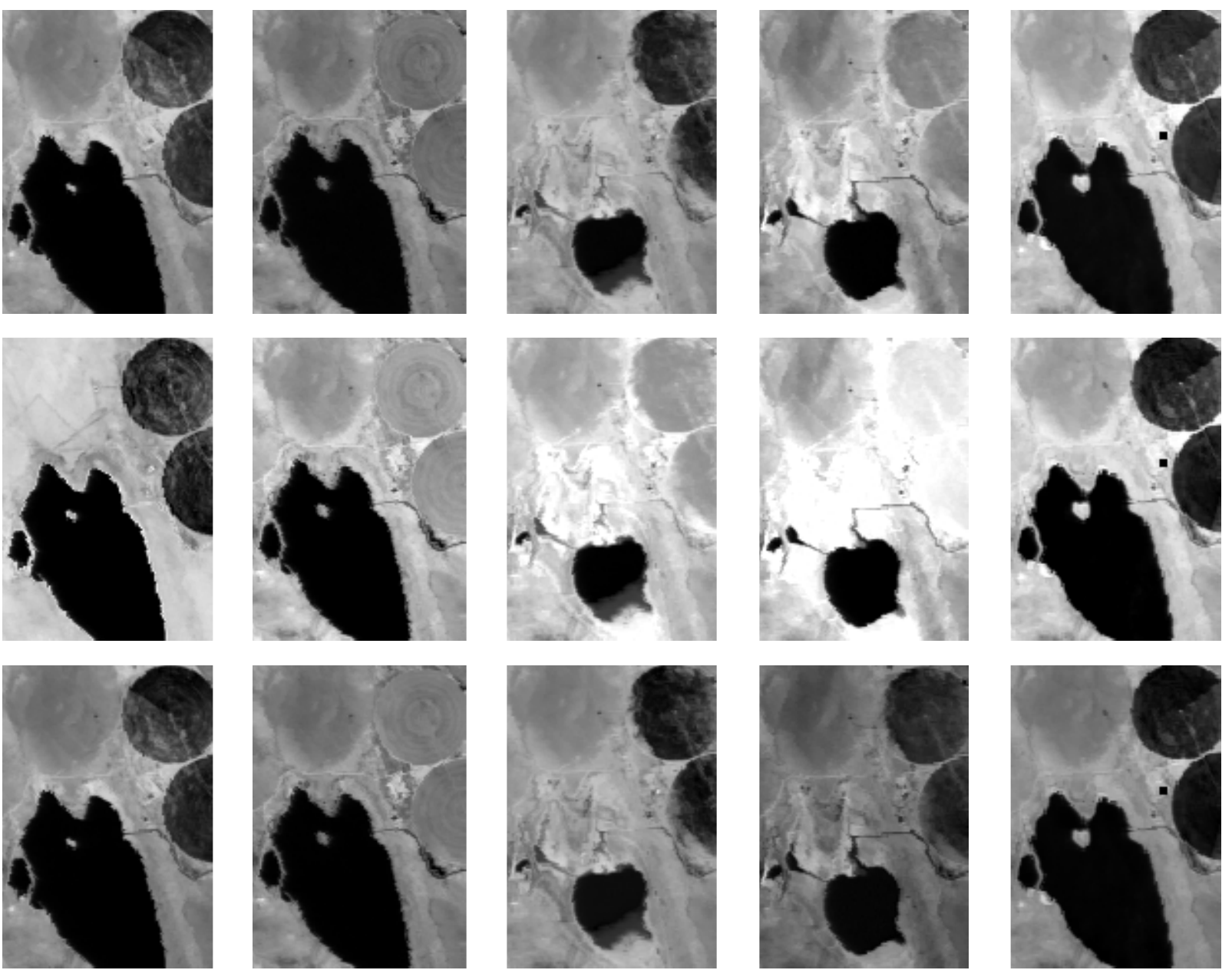}
\caption{Soil abundance maps (proposed method on the first line, VCA/FCLS on the second line, SISAL/FCLS on the third line).}
\label{fig:abundance2} \vspace{-0.5cm}
\end{figure}

\begin{figure}[t]
\centering
\includegraphics[keepaspectratio,height=0.2\textheight , width=0.5\textwidth]{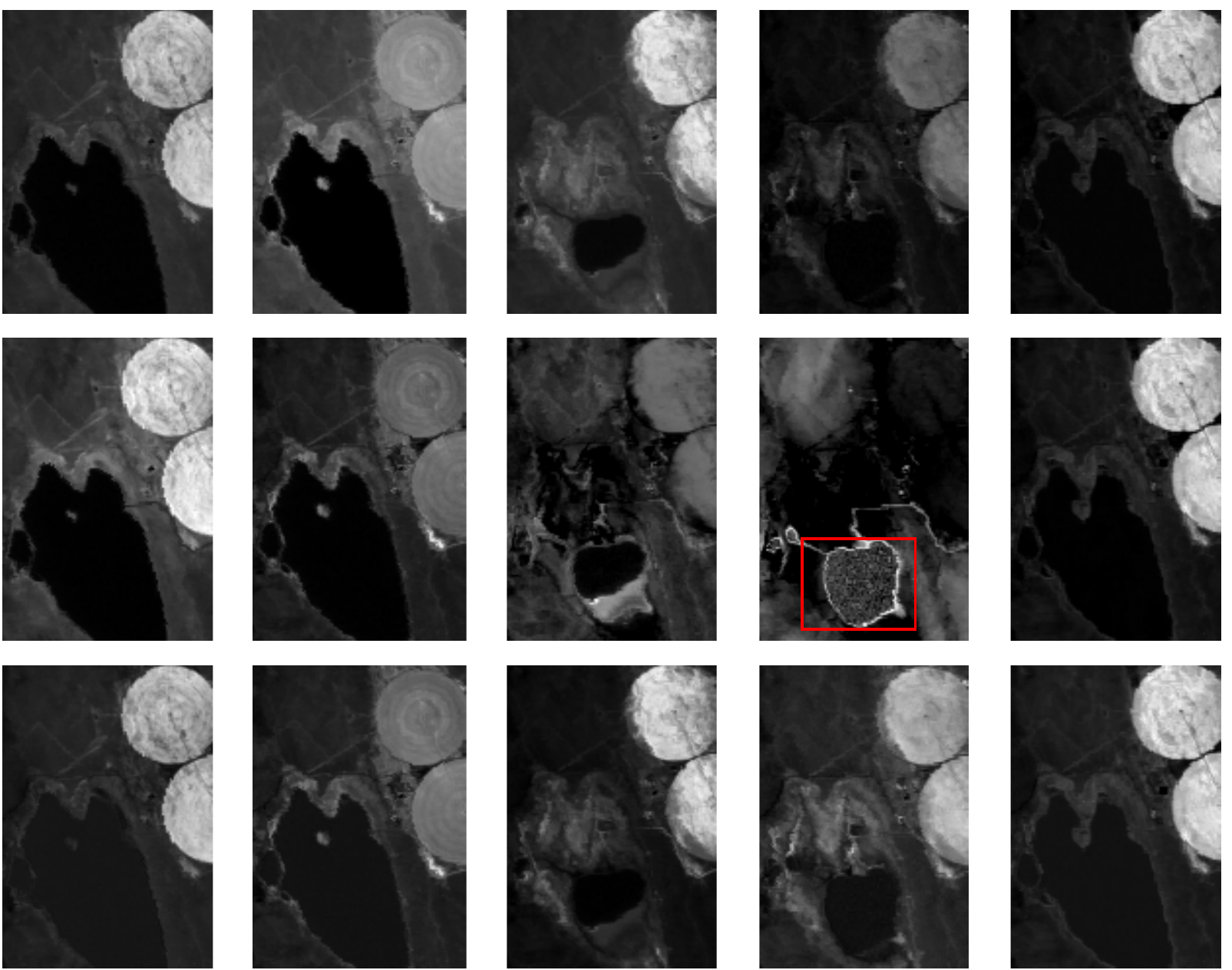}
\caption{Vegetation abundance maps (proposed method on the first line, VCA/FCLS on the second line, SISAL/FCLS on the third line). The region delineated in red, where almost no vegetation is supposed to be present, reveals that the water endmember extracted by VCA has been split into two parts. This observation is further confirmed in Figs. \ref{fig:vca_endm1} and \ref{fig:vca_endm3}.}
\label{fig:abundance3}
\end{figure}	

	\subsection{Results}
	Since no ground truth is available, the algorithm performance is evaluated in terms of the reconstruction error defined in \eqref{eq:RE}. Only the more consistent abundance maps and endmembers obtained for $\nendm = 3$ are presented in Figs. \ref{fig:abundance1} to \ref{fig:real_endm} due to space constraints. Complementary results are available in \cite{Thouvenin2015TR}. The proposed method provides comparable reconstruction errors (see Table \ref{tab:results_real}), yields more consistent abundance maps when compared to VCA/FCLS and SISAL/FCLS especially for the soil and the vegetation for a somewhat reasonable computational cost. In particular, note that the estimated vegetation abundance map of the fourth image depicted in Fig. \ref{fig:abundance3} (area delineated in red) presents significant errors when visually compared to the corresponding RGB image in Fig. \ref{fig:cube4}. These errors can be explained by the fact that the water endmember extracted by VCA has been split into two parts as can be seen in Figs. \ref{fig:vca_endm1} and \ref{fig:vca_endm3} (see signatures given in black). Indeed, the VCA algorithm cannot detect the scarcely present vegetation. On the contrary, the joint exploitation of multiple images enables the faint traces of dry vegetation to be captured. Albeit impacted by the results of VCA/FCLS (used as initialization), the performance of $\ell_{1/2}$ NMF and BCD/ADMM remains satisfactory on each image of the sequence since they tend to correct the endmember errors induced by VCA. However, $\ell_{1/2}$ NMF produces undesirable endmembers with an amplitude significantly greater than 1 on the 4th image (Fig. \ref{fig:cube4}). Besides, BCD/ADMM yields very low reconstruction errors at the price of a computational cost which may become prohibitive for extended image sequences. The figures related to $\ell_{1/2}$ NMF and BCD/ADMM are available in the associated report \cite{Thouvenin2015TR} due to space constraints.
	
Furthermore the instantaneous variability energy (computed as $\norm2{\mathbf{dm}_{rt}}/ \nband$ for $r=1,\dotsc,\nendm$ and $t=1\dotsc,\ntime$) can reveal which endmember deviates the most from its average spectral behavior. In this experiment, the soil and the vegetation signatures -- which seem to vary the most over time (see Fig. \ref{fig:cube}) -- are found by the proposed method to be affected by the most significant variability level (see Table \ref{tab:results_var_norm}). In this experiment, a significant increase can be observed in the endmember variability energy over the last three images of the sequence (see Table \ref{tab:results_var_norm}), suggesting that the endmembers are apparently better represented in the two first images of the sequence (see Fig. \ref{fig:cube}). This observation suggests the proposed method captures the average endmember spectral behavior and enables the time at which the greatest spectral changes occur to be identified. However, a detailed analysis of this observation is out of the scope of the present paper.	

\begin{figure}[t]
\centering
\subfloat[4][Water]{
\includegraphics[keepaspectratio,height=0.15\textheight , width=0.15\textwidth]{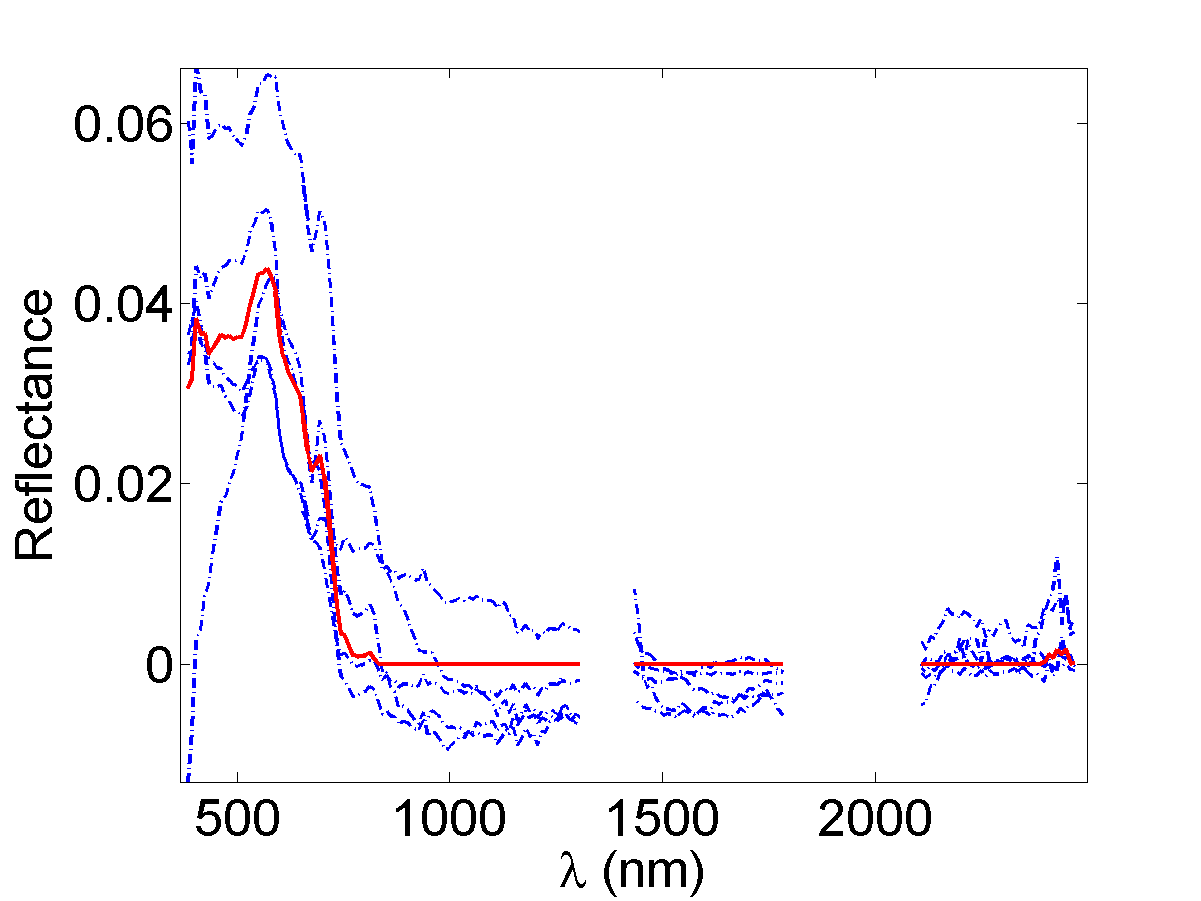}
\label{fig:endm1}}
\subfloat[5][Soil]{
\includegraphics[keepaspectratio,height=0.15\textheight , width=0.15\textwidth]{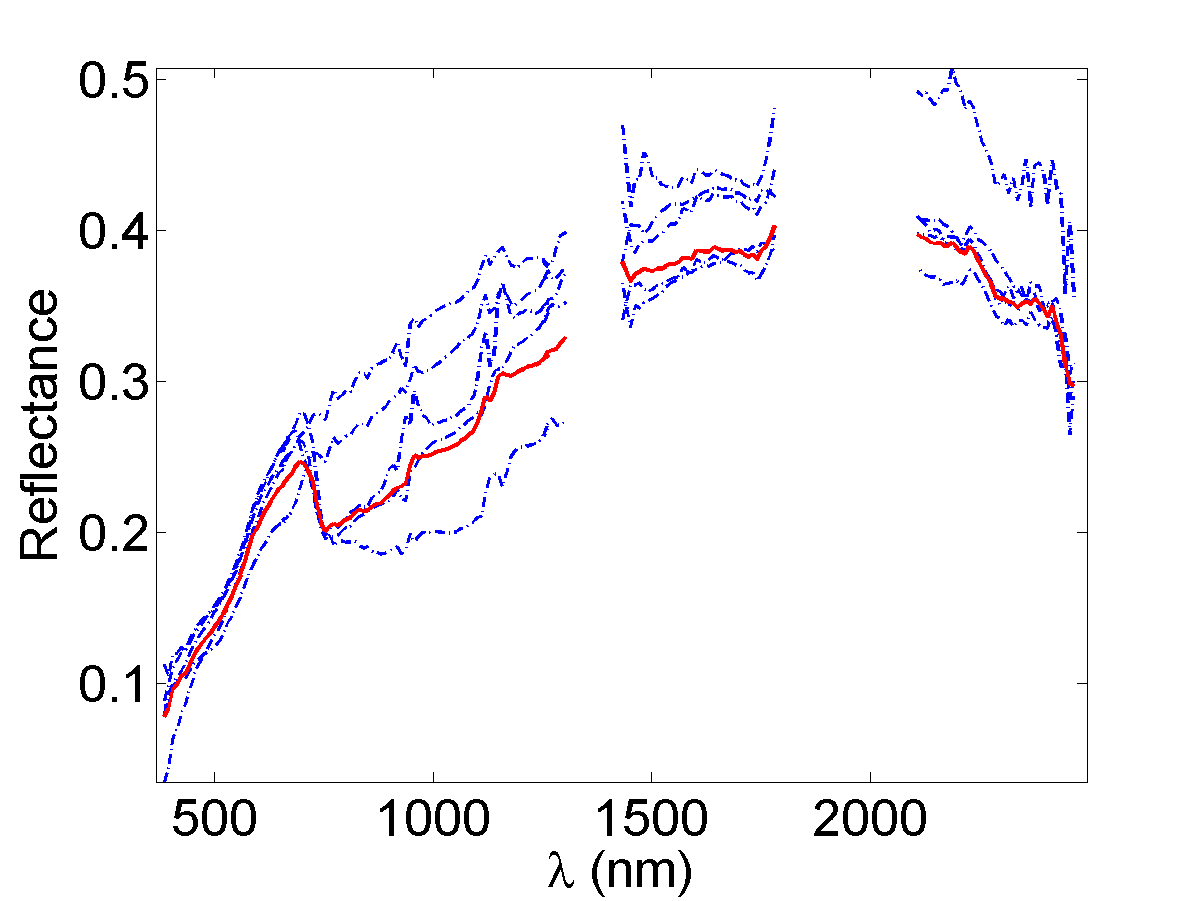}
\label{fig:endm2}}
\subfloat[6][Vegetation]{
\includegraphics[keepaspectratio,height=0.15\textheight , width=0.15\textwidth]{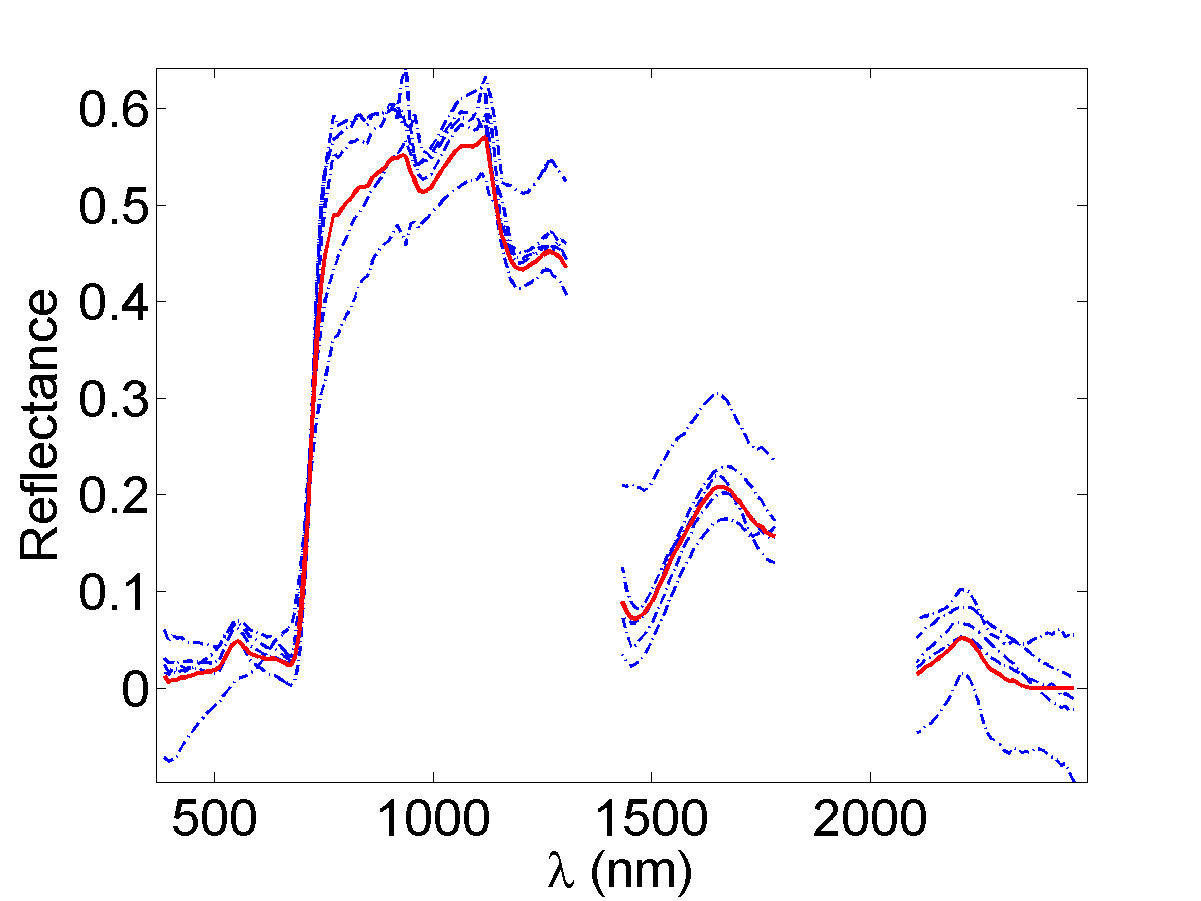}
\label{fig:endm3}}
\vspace{-0.25cm}
\\
\subfloat[4][Water (VCA)]{
\includegraphics[keepaspectratio,height=0.15\textheight , width=0.15\textwidth]{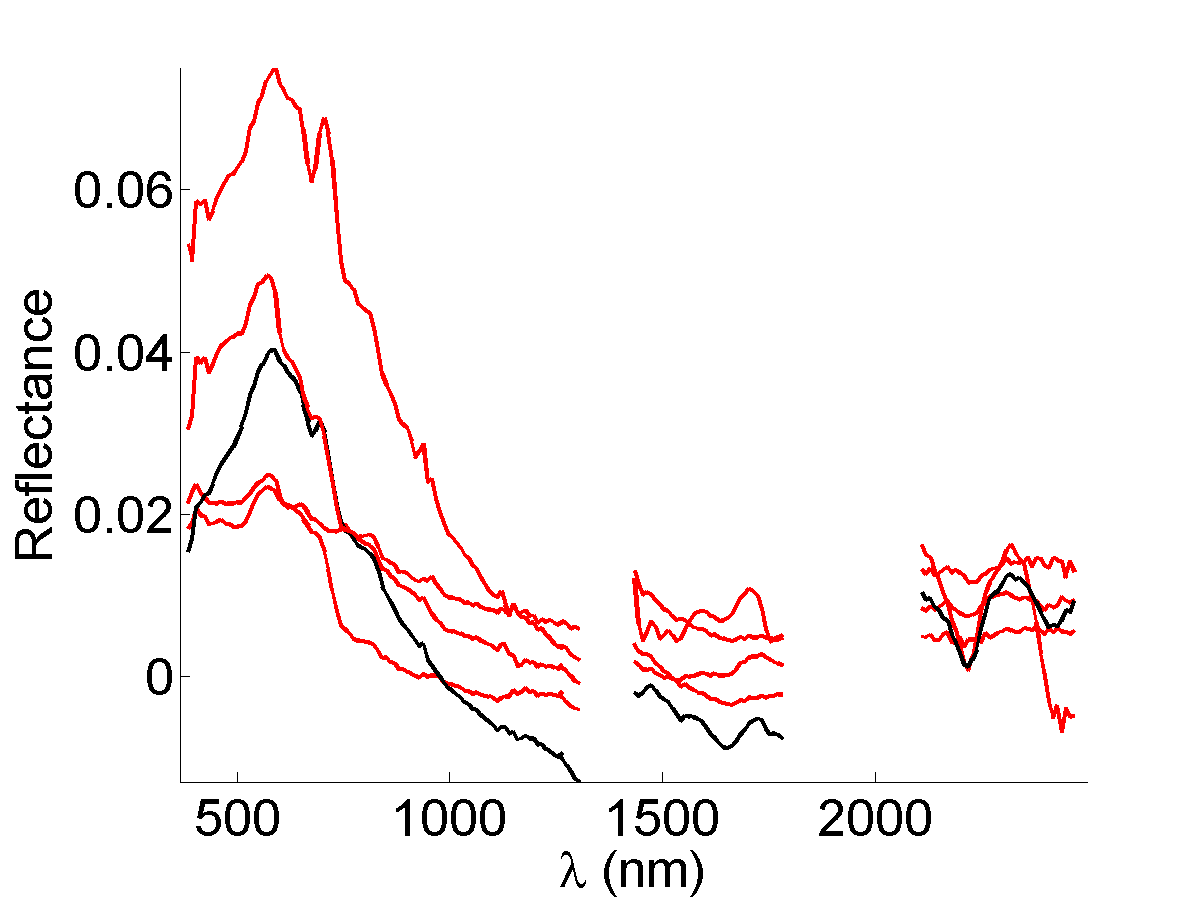}
\label{fig:vca_endm1}}
\subfloat[5][Soil (VCA)]{
\includegraphics[keepaspectratio,height=0.15\textheight , width=0.15\textwidth]{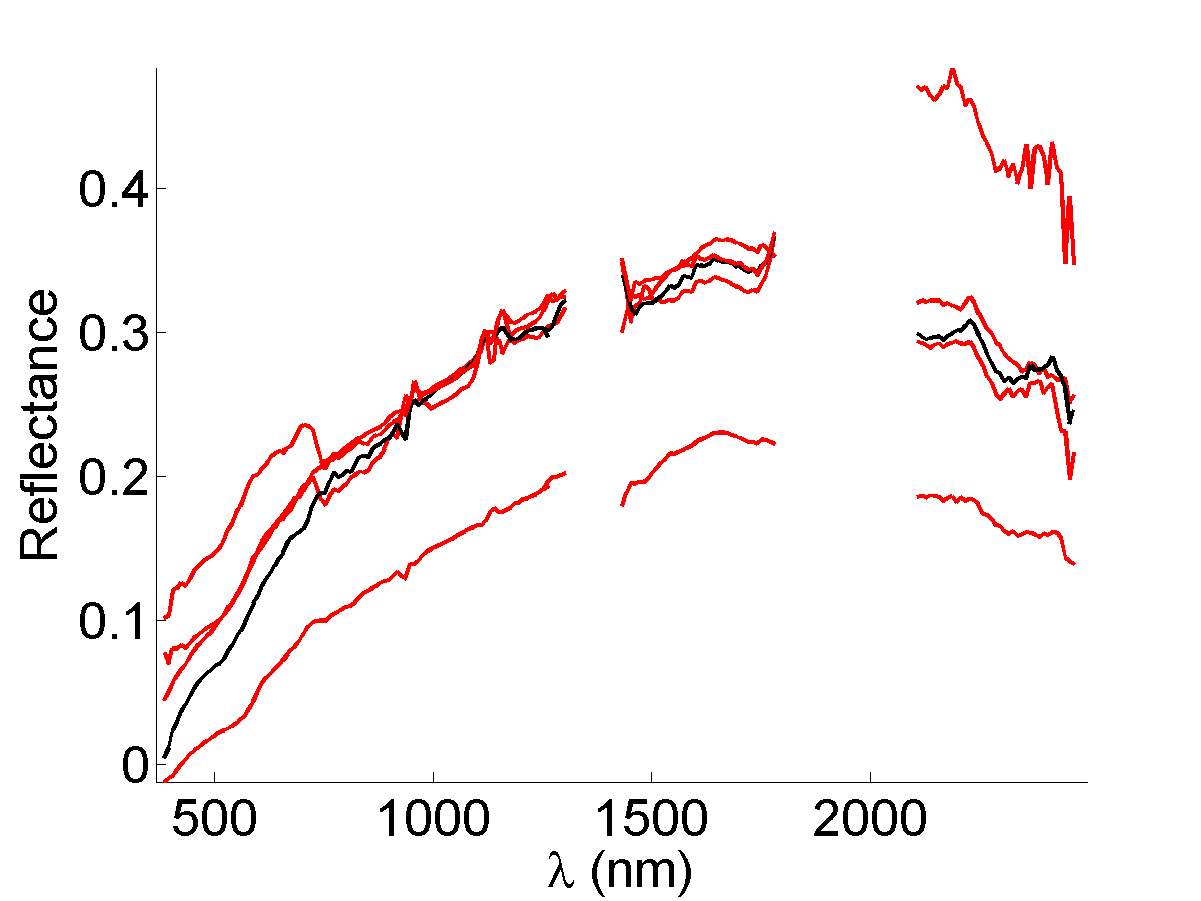}
\label{fig:vca_endm2}}
\subfloat[6][Vegetation (VCA)]{
\includegraphics[keepaspectratio,height=0.15\textheight , width=0.15\textwidth]{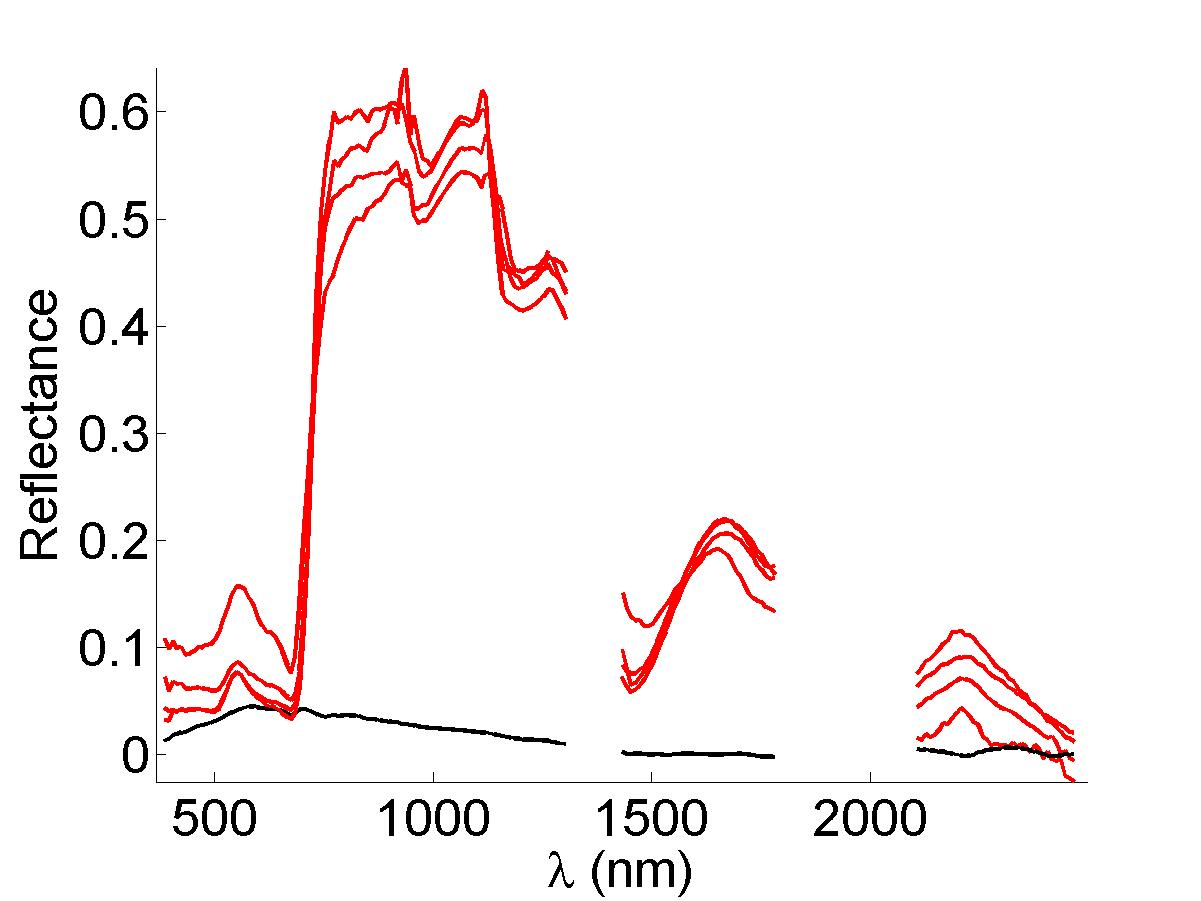}
\label{fig:vca_endm3}}
\\
\subfloat[7][Water (SISAL)]{
\includegraphics[keepaspectratio,height=0.15\textheight , width=0.15\textwidth]{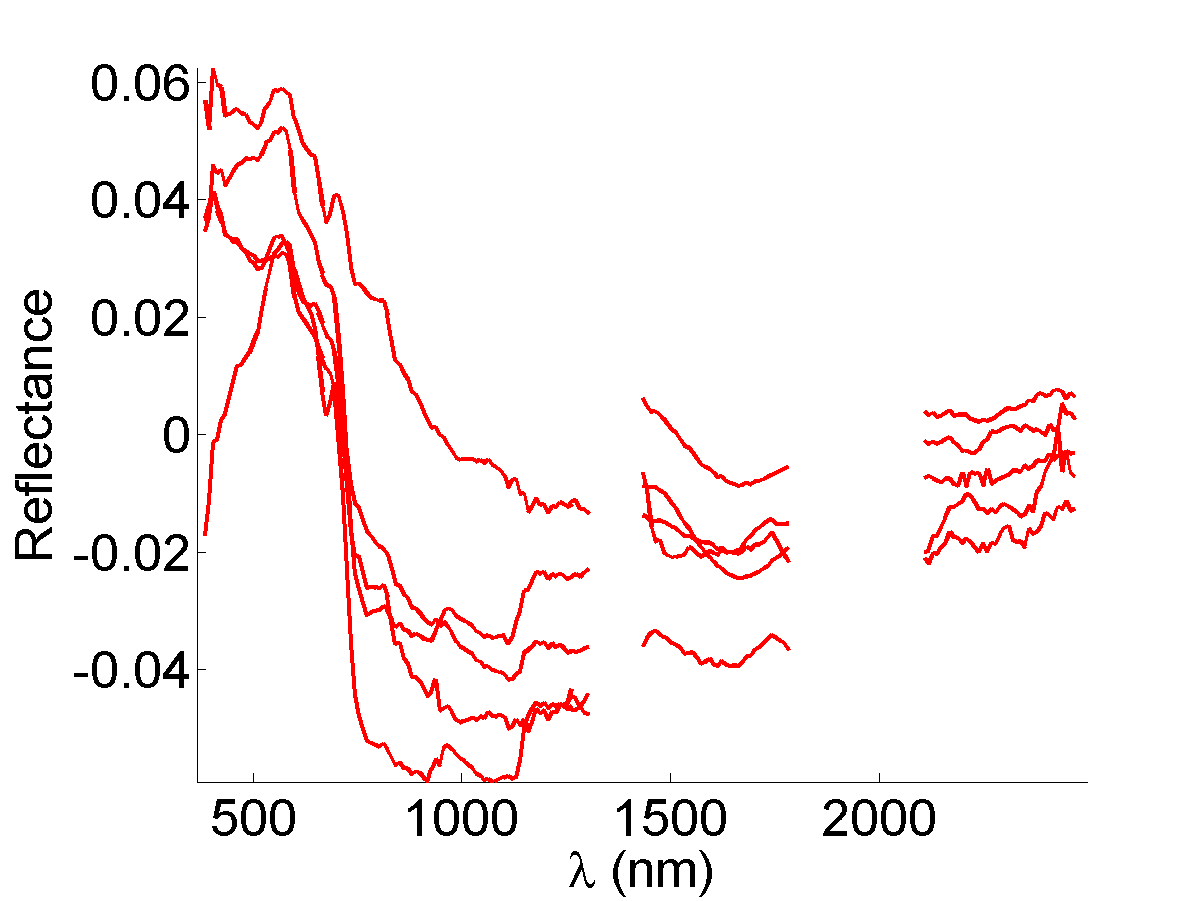}
\label{fig:sisal_endm1}}
\subfloat[8][Soil (SISAL)]{
\includegraphics[keepaspectratio,height=0.15\textheight , width=0.15\textwidth]{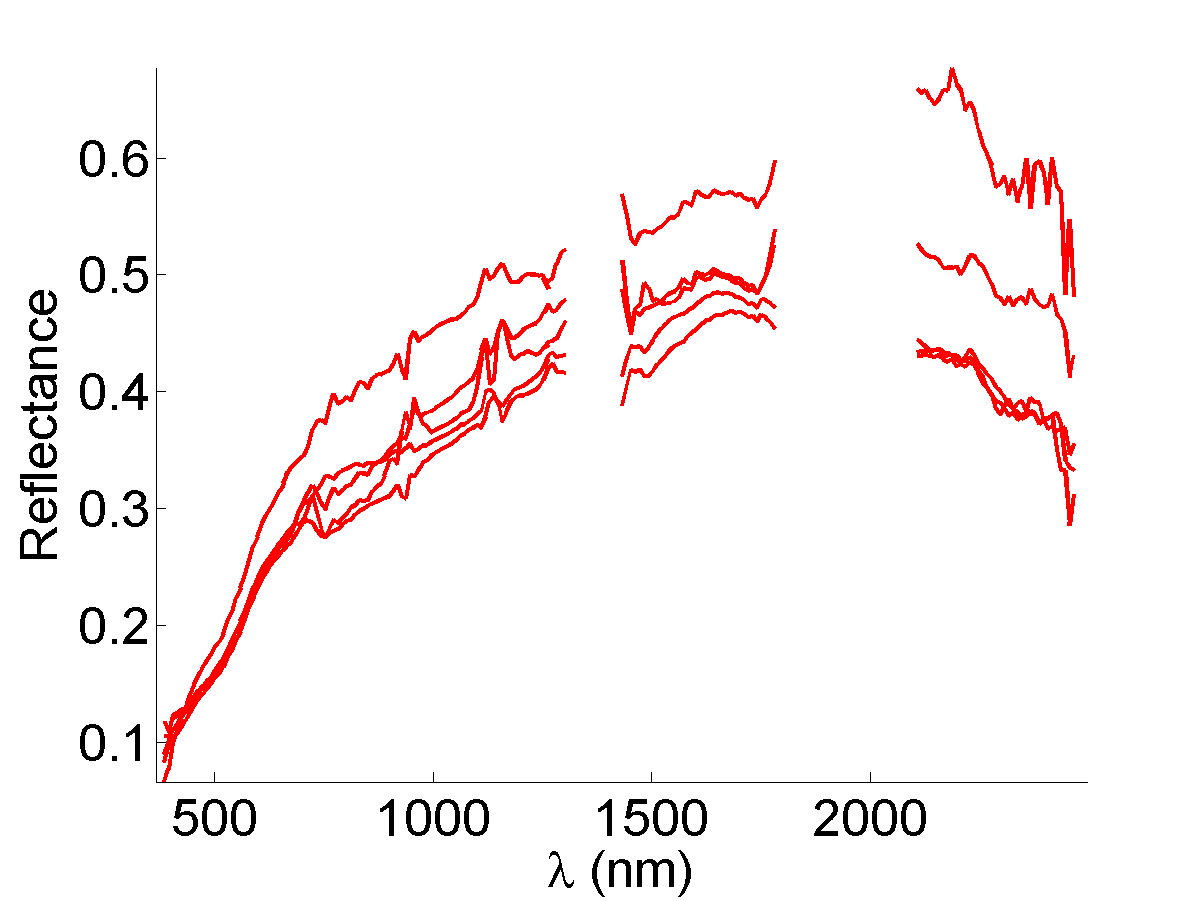}
\label{fig:sisal_endm2}}
\subfloat[9][Vegetation (SISAL)]{
\includegraphics[keepaspectratio,height=0.15\textheight , width=0.15\textwidth]{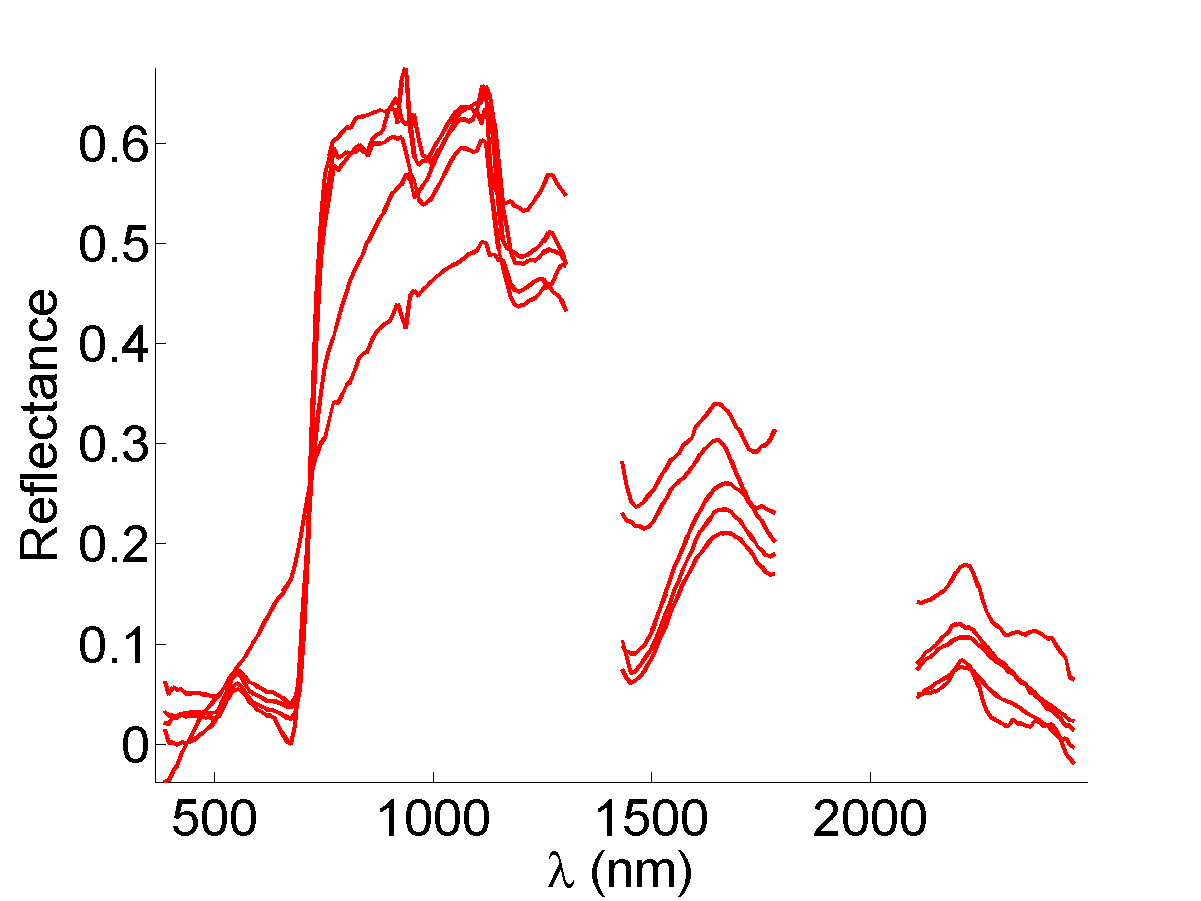}
\label{fig:sisal_endm3}}
\caption{Endmembers and variability (endmembers in red lines, variability in blue dashed lines) recovered by the proposed method on the first line, VCA-extracted endmembers on the second line, SISAL-extracted endmembers on the third line). The endmembers given in black on the second line correspond to the endmembers identified by VCA on the image \ref{fig:cube4}, where the water endmember has been split into two parts (see Figs. \ref{fig:vca_endm1} and \ref{fig:vca_endm3}).}
\label{fig:real_endm}
\end{figure}
\begin{table}[t]
\vspace{-0.5cm}
\caption{Endmember number $\nendm$ estimated on each image of the real dataset by HySime \cite{Bioucas2008} and EGA \cite{Halimi2016}.}
	\begin{center}
		\begin{tabular}{@{}lccccc@{}} \toprule
			& \rotatebox{45}{04/10/2014} & \rotatebox{45}{06/02/2014} & \rotatebox{45}{09/19/2014} & \rotatebox{45}{11/17/2014} & \rotatebox{45}{04/29/2015}  \\ \cmidrule{2-6}
HySime \cite{Bioucas2008}	& 16 & 21 & 19 & 21 & 22 \\
EGA \cite{Halimi2016}       & 3 & 5 & 4 & 3 & 3 \\ \bottomrule	
		\end{tabular}
	\end{center}
\label{tab:results_r} \vspace{-0.3cm}
\end{table}
\setlength\columnsep{0.1pt}
\begin{table}[!t] 
\vspace{-0.3cm}
\caption{Simulation results on real data (RE $\times 10^{-4}$).}
	\begin{center}
		\begin{tabular}{@{}llcc@{}} \toprule
&		   	   & RE & time (s) \\ \cmidrule{1-4}
\multirow{5}{*}{\rotatebox{90}{$\nendm = 3$}}
&VCA/FCLS      &  12.7 & \textbf{2} \\
&SISAL/FCLS	   &  0.87 & 3 	\\
&$\ell_{1/2}$ NMF & 3.83 & 156 \\
&BCD/ADMM	  & \textbf{0.37} & 2449 \\
&Proposed 	  & 1.04 & 134 \\ \cmidrule{1-4}
\multirow{5}{*}{\rotatebox{90}{$\nendm = 4$}}
&VCA/FCLS      &  43.8 & \textbf{2} \\
&SISAL/FCLS	   &  0.35 & 3 	\\
&$\ell_{1/2}$ NMF & 16.0 & 163 \\
&BCD/ADMM	      & \textbf{0.27} & 4396 \\
&Proposed 	  & 0.76 & 126 \\ \cmidrule{1-4}
\multirow{5}{*}{\rotatebox{90}{$\nendm = 5$}}
&VCA/FCLS      &  63.9 & \textbf{2} \\
&SISAL/FCLS	   &  0.17 & 4 	\\
&$\ell_{1/2}$ NMF & 14.6 & 174 \\
&BCD/ADMM	      & \textbf{0.098} & 12511 \\
&Proposed 	  & 0.17 & 128 \\
    \bottomrule
		\end{tabular}
	\end{center}
\label{tab:results_real} \vspace{-0.5cm}
\end{table}
\begin{table}[t]
\vspace{-0.5cm}
\caption{Experiment with real data for $\nendm = 3$: energy of the variability captured for each endmember at each time instant ($\norm2{\mathbf{dm}_{kt}}/\nband \times 10^{-5}$ for $k = 1,\dotsc,\nendm$, $t = 1,\dotsc ,\ntime$).}
	\begin{center}
		\begin{tabular}{@{}lccc@{}} \toprule
			& Water	& Vegetation & Soil \\ \cmidrule{2-4}
04/10/2014	& 1.22	& 9.68		    & \textbf{11.51} \\
06/02/2014	& 1.44	& 11.85         & \textbf{38.37} \\
09/19/2014	& 7.29	& \textbf{11.41}& 9.30 \\
11/17/2014	& 2.77	& \textbf{21.73}& 16.55 \\
04/29/2015	& 0.58	& \textbf{106.03}& 26.19 \\ \bottomrule
		\end{tabular}
	\end{center}
\label{tab:results_var_norm} \vspace{-0.8cm}
\end{table}

\section{Conclusion and future work} \label{sec:conclusion}
This paper introduced an online hyperspectral unmixing procedure accounting for endmember temporal variability based on the perturbed linear model considered in \cite{Thouvenin2015}. This algorithm was designed to unmix multiple HS images of moderate size, potentially affected by smoothly varying endmember perturbations. Indeed, the number of spurious local optima of the cost function used in this paper can significantly increase with the size of the images and the number of endmembers considered, which is a problem common to many blind source separation problems (such as the unmixing problem addressed in this paper). The underlying unmixing problem was formulated as a two-stage stochastic program solved by a stochastic approximation algorithm. Simulations conducted on synthetic and real data enabled the interest of the proposed approach to be appreciated. Indeed, the proposed method compared favorably with established approaches performed independently on each image of the sequence while providing a relevant variability estimation. Assessing the robustness of the proposed technique with respect to estimation errors on the endmember number $\nendm$ and applying the proposed method to real dataset composed of a larger number of endmembers are interesting prospects for future work. Possible perspectives also include the extension of the method to account for spatial variability and applications to change detection problems. A distributed unmixing procedure is also under investigation to solve the resulting high dimensional problem.

\begin{appendices}
\section{Projections involved in the parameter updates} \label{sec:proj}
The projections involved in the PALM algorithm \cite{Bolte2013} described in Algo. \ref{alg:update_A_dM} are properly defined, since the associated constraint spaces are closed convex sets. More precisely,
\begin{itemize}
\item $\D_t$ is closed and convex as the (non-empty) intersection of two closed balls. The projection onto $\D_t$ can be approximated by the Dykstra algorithm \cite{Boyle1986,Heylen2013}. Besides, the projection on a Frobenius ball is given by \cite{Parikh2014} \vspace{-0.15cm}
	\begin{equation}
	\mathcal{P}_{\Ball{\mathbf{X}}{r}} (\Y) = \mathbf{X} + \min \left(1, \frac{r}{\normf{\Y - \mathbf{X}}} \right) (\Y - \mathbf{X});\vspace{-0.2cm}
	\end{equation}
\item projecting $\M$ onto $\R{\nband}{\nendm}_+$ is explicitly given by \vspace{-0.2cm}
\begin{equation}
\mathcal{P}_+ (\M) = \max(\Zero{\nband}{\nendm}, \M) \vspace{-0.2cm}
\end{equation}
where the $\max$ is taken term-wise.
\end{itemize}

\section{Discussion on Assumption \ref{hyp:hessian}} \label{sec:hessian}
The Hessian matrix of $f(\Y,\M,\cdot,\cdot)$, denoted by $\Hf$, is given by \vspace{-0.2cm}
\begin{align}
&\Hf = \begin{bmatrix}
\mathbf{H}_1 & \mathbf{H}_2 \\
\mathbf{H}_3 & \mathbf{H}_4
\end{bmatrix} \\
&\Mt = (\M+\dM) \\
&\mathbf{H}_1 = \I_\nbpix \otimes (\t{\Mt}\Mt), \quad \mathbf{H}_4 = (\A \t{\A}) \otimes \I_\nband \\
&\mathbf{H}_3 = \t{\mathbf{H}_2} = \biggl\{ \I_\nendm \otimes \bigl[-\Y + \Mt\A \bigr] \biggr\}\shuffle{\nendm,\nbpix} + [\A \otimes \Mt]. \vspace{-0.2cm}
\end{align}
where $\shuffle{\nendm,\nband}$ is the perfect shuffle matrix. The block matrix $\Hf$ is invertible if, for instance, $\mathbf{H}_1$ and its Schur complement $\mathbf{S} = \mathbf{H}_4 - \mathbf{H}_3 \mathbf{H}_1^{-1}\mathbf{H}_2$ are invertible. In practice, $\mathbf{H}_1$ is generally invertible since $\M + \dM$ is full column rank. The invertibility of the Schur complement $\mathbf{S}$ can be ensured via an appropriate regularization term $\frac{\mu}{2} \normF2{\A}$ added to the original objective $f$. Indeed, we first note that such a perturbation regularizes the Hessian by modifying its diagonal block $\mathbf{H}_4$, replaced by $\mathbf{H}_4 + \mu \mathbf{I}$.

Denote by $\lambda_1 > \lambda_2 > \dotsc > \lambda_r$ the ordered eigenvalues of $\mathbf{S}$, where $r$ denotes the number of distinct eigenvalues. By the spectral theorem, there exists an orthogonal matrix (\wrt the canonical euclidean inner product) $\mathbf{Q}$ such that $\mathbf{S} = \t{\mathbf{Q}}\mathbf{DQ}$, where $\mathbf{D}$ is a diagonal matrix composed of the $\lambda_k$. Note that each eigenvalue may have a multiplicity order greater than 1 with the adopted notations. If there exits $k$ such that $\lambda_k = 0$, then $\lambda_{k+1} < 0$. Adding $\frac{\mu}{2} \normF2{\A}$ to the original objective function, with $\mu < |\lambda_{k+1}|$, is then sufficient to ensure the invertibility of the Schur complement \vspace{-0.2cm}
\begin{equation*}
(\mathbf{H}_4 - \mathbf{H}_3 \mathbf{H}_1^{-1}\mathbf{H}_2) + \mu \mathbf{I} = \t{\mathbf{Q}}\mathbf{DQ} + \mu \mathbf{I} =  \t{\mathbf{Q}}(\mathbf{D}+\mu \mathbf{I})\mathbf{Q} \vspace{-0.2cm}
\end{equation*} 
associated to the new Hessian matrix, thus ensuring its invertibility.

\section{Convergence proof} \label{sec:proof}

Largely adapted from \cite{Mairal2010}, the following sketch of proof reduces to an adaptation of \cite[Lemma 1, Proposition 1]{Mairal2010}. From this point, our problem exactly satisfies the assumptions required to apply the same arguments as in \cite[Proposition 2, Proposition 3]{Mairal2010}, leading to the announced convergence result.

\begin{lemma}[Asymptotic variations of $\M_t$ \cite{Mairal2010}]
Under Assumptions \ref{hyp:gt} and \ref{hyp:Lipschitz}, we have \vspace{-0.2cm}
\begin{equation}
\normf{\M^{(t+1)} - \M^{(t)}} = O\left( \frac{1}{t} \right) \; \text{almost surely} \; (a.s.). \vspace{-0.2cm}
\end{equation}
\end{lemma}
\begin{proof}
According to Assumption \ref{hyp:gt}, $\gt$ is strictly convex with a Hessian lower-bounded by a scalar $\mu_\M > 0$. Consequently, $\gt$ satisfies the second-order growth condition \vspace{-0.2cm}
	\begin{equation}
	\label{eq:lemma1_growth}
	\gt(\M^{(t+1)}) - \gt(\M^{(t)}) \geq \mu_M \normF2{\M^{(t+1)} - \M^{(t)}}. \vspace{-0.2cm}
	\end{equation}
Besides, since $\M \in [0,1]^{\nband \times \nendm}$, we have $\normf{\M} \leq \sqrt{\nband \nendm}$. Hence $\gt$ is Lipschitz continuous with constant $c_t = \frac{1}{t} \left( \normf{\mathbf{D}_t} + \sqrt{\nband \nendm} \normf{\mathbf{C}_t} \right) + \beta c_\Psi$. Indeed, given two matrices $\M_1,\M_2 \in [0,1]^{\nband \times \nendm}$, we have \vspace{-0.2cm}
	\begin{equation}
\begin{split}
&|\gt(\M_1) - \gt(\M_2)| \leq \beta \bigl| \Psi(\M_1) - \Psi(\M_2) \bigr| + \\
& \frac{1}{t} \left| \frac{1}{2} \langle \t{\M_1}\M_1 - \t{\M_2}\M_2, \mathbf{C}_t \rangle - \langle \M_1 - \M_2, \mathbf{D}_t \rangle \right| \\
& \leq \beta c_\Psi \normf{\M_1 - \M_2} + \frac{1}{t}\normf{\M_1 - \M_2}\normf{\mathbf{D}_t} \\
& + \frac{1}{2t} \normf{\t{\M_1}\M_1 - \t{\M_2}\M_2}\normf{\mathbf{C}_t}
\end{split} \vspace{-0.2cm}
	\end{equation}
where $\mathbf{C}_t$ and $\mathbf{D}_t$ were defined in \eqref{eq:CD}.
In addition \vspace{-0.2cm}
	\begin{equation}
\begin{split}
&\normf{\t{\M_1}\M_1 - \t{\M_2}\M_2} = \frac{1}{2} \Vert \t{(\M_1 + \M_2)}(\M_1 - \M_2) \\
& + \t{(\M_1 - \M_2)}(\M_1 + \M_2) \Vert_{\text{F}} \\
%
& \leq 2 \sqrt{\nband \nendm}\normf{\M_1 - \M_2}
\end{split} \vspace{-0.2cm}
	\end{equation}
hence \vspace{-0.3cm}
	\begin{equation}
\label{eq:lemma1_lipschitz}
\left| \gt(\M_1) - \gt(\M_2) \right| \leq c_t \normf{\M_1 - \M_2}. \vspace{-0.2cm}
	\end{equation}
Combining \eqref{eq:lemma1_growth} and \eqref{eq:lemma1_lipschitz}, we have \vspace{-0.2cm}
\begin{equation}
\normf{\M^{(t+1)} - \M^{(t)}} \leq \frac{c_t}{\mu_\M}. \vspace{-0.2cm}
\end{equation}
Since the data, the abundances and the variability are respectively contained in compact sets, $\mathbf{C}_t$ and $\mathbf{D}_t$ are (almost surely) bounded, thus: $c_t = O \left( \frac{1}{t} \right)$ a.s.
\end{proof}

\begin{proposition}[Adapted from \cite{Mairal2010}]
We assume that the requirements in Assumption \ref{hyp:gt} to \ref{hyp:hessian} are satisfied. Let $(\Y_t,\M)$ be an element of $\mathcal{Y}~\times~\mathcal{M}$. Let us define \vspace{-0.2cm}
	\begin{align}
	\mathcal{Z}_t =& \Ak \times \mathcal{D}_t \\ 
	\mathcal{Q}(\Y_t,\M) =& \{(\A,\dM) \in \mathcal{Z}_t | \nonumber \\ &\nabla_{(\A,\dM)} f(\Y_t,\M,\A,\dM) = \mathbf{0} \} \\
	(\A_t^*,\dM_t^*) \in &\mathcal{Q}(\Y_t,\M) \\	
	v(\Y_t,\M)  =& f \bigl( \Y_t,\M,\A_t^*,\dM_t^* \bigr). \vspace{-0.3cm}
	\end{align}
Then  \vspace{-0.1cm}
	\begin{enumerate}
	\item the function $v$ is continuously differentiable \wrt{} $\M$ and $\nabla_\M v(\Y_t,\M) = \nabla_\M f \bigl( \Y_t,\M,\A_t^*,\dM_t^* \bigr)$;
	\item $g$ defined in \eqref{eq:problem} is continously differentiable and $\nabla_\M g(\M) = \mathbb{E}_{\Y_t} \bigl[ \nabla_\M z(\Y_t,\M) \bigr]$;
	\item $\nabla_\M g$ is Lipschitz continuous on $\mathcal{M}$. \vspace{-0.1cm}
	\end{enumerate}
\end{proposition}
\begin{proof}
The existence of local minima of $f(\Y_t,\M,\cdot,\cdot)$ on $\mathcal{Z}_t$ follows from the continuity of $f(\Y_t,\M,\cdot,\cdot)$ and the compactness of $\mathcal{Z}_t$. This ensures the non-emptiness of $\mathcal{Q}(\Y_t,\M)$ and justifies the definition of $(\A_t^*,\dM_t^*)$.

Furthermore, Assumption \ref{hyp:hessian} requires the invertibility of the Hessian matrix $\Hf$ at the point $(\Y_t,\M,(\A_t^*,\dM_t^*))$. The first statement then follows from the implicit function theorem \cite[Theorem 5.9 p.19]{Lang1999}: there exist two open subsets $V \subset \mathcal{M}$,  $W \subset \mathcal{Z}_t$ and a continuously differentiable function $\varphi : V \longrightarrow W $ such that
\begin{enumerate}[label=(\roman*)]
\item $(\M,(\A_t^*,\dM_t^*)) \in V \times W \subset \mathcal{M} \times \mathcal{Z}_t$;
\item for all $(\Mtilde,(\A,\dM)) \in V \times W$, we have
\end{enumerate}
\begin{equation} \label{eq:phi}
\begin{split}
[\nabla_{(\A,\dM)} f(\Y_t,\Mtilde,&\A,\dM) = \mathbf{0}] \\ 
&\Rightarrow [(\A,\dM) = \varphi(\Mtilde)];
\end{split} \vspace{-0.3cm}
\end{equation}
\begin{enumerate}[label=(\roman*),resume]
\item for all $\Mtilde \in V$,
\end{enumerate}
\begin{equation}
\begin{split}
\frac{\partial \varphi}{\partial \M} (\Mtilde) = &- \mathbf{H}_{(\A,\dM)}^{-1}f (\Y_t,\Mtilde,\varphi(\Mtilde)) \\ &\frac{\partial f}{\partial \M \partial (\A,\dM)} (\Y_t,\Mtilde, \varphi(\Mtilde)).
\end{split} \vspace{-0.3cm}
\end{equation}

In particular, $(\M,(\A_t^*,\dM_t^*)) \in V \times W$ satisfies \eqref{eq:phi}. Then, taking the derivative of $v( \Y_t,\cdot )$ in $\M$ leads to \vspace{-0.2cm}
\begin{equation}
\begin{split}
\frac{\partial v}{\partial \M} (\Y_t,\M) = &\underbrace{\frac{\partial f}{\partial (\A,\dM)} (\Y_t,\M,\varphi(\M))}_{ = \mathbf{0} \text{ since } \varphi(\M) \in \mathcal{Q}(\Y_t,\M)} \frac{\partial \varphi}{\partial \M} (\M) \\
&+ \frac{\partial f}{\partial \M}(\Y_t,\M,\varphi(\M))
\end{split} \vspace{-0.2cm}
\end{equation}
%
%

The second statement follows from the continuous differentiability of $z(\Y_t,\cdot)$, defined on a compact set.

We finally observe that $\normf{\A_t^*}$ and $\normf{\dM_t^*}$ are respectively bounded by a constant independent from $\Y_t$ (since $(\A_t^*,\dM_t^*) \in \Ak \times \mathcal{D}_t$). This observation, combined with the expression of $\nabla_\M f$ and the compactness of $\mathcal{M}$, leads to the third claim.
\end{proof}
\end{appendices}


\bibliographystyle{IEEEtran}
\bibliography{strings_all_ref2,all_ref2}
\vspace{-1.2cm}


\begin{IEEEbiographynophoto}{Pierre-Antoine Thouvenin} (S'15) received the state engineering degree in electrical engineering from ENSEEIHT, Toulouse, France, and the M.Sc. degree in signal processing from the National Polytechnic Institute of Toulouse (INP Toulouse), both in 2014. He is currently working toward the Ph.D. degree within the Signal and Communications Group of the IRIT Laboratory, Toulouse, France. His research is currently focused on hyperspectral unmixing and variability modeling in hyperspectral imagery.
\end{IEEEbiographynophoto}

\vspace{-1.1cm}

\begin{IEEEbiographynophoto}{Nicolas Dobigeon} (S'05--SM'08--SM'13) received the state engineering degree in electrical engineering from ENSEEIHT, Toulouse, France, and the M.Sc. degree in signal
processing from the National Polytechnic Institute of Toulouse (INP Toulouse), both in June 2004, as well as the Ph.D. degree and Habilitation \`a Diriger des Recherches in Signal Processing from the INP Toulouse in 2007 and 2012, respectively.
He was a Post-Doctoral Research Associate with the Department of Electrical Engineering and Computer Science, University of Michigan, Ann Arbor, MI, USA, from 2007 to 2008. 

Since 2008, he has been with the National Polytechnic Institute of Toulouse (INP-ENSEEIHT, University of Toulouse) where he is currently an Associate Professor. He conducts his research within the Signal and Communications Group of the IRIT Laboratory and he is also an affiliated faculty member of the Telecommunications for Space and Aeronautics (TeSA) cooperative laboratory.
His current research interests include statistical signal and image processing, with a particular interest in Bayesian inverse problems with applications to remote sensing, biomedical imaging and genomics.
\end{IEEEbiographynophoto}

\vspace{-1cm}

\begin{IEEEbiographynophoto}{Jean-Yves Tourneret} (SM'08) received the ing\'enieur degree in electrical engineering from the Ecole Nationale Sup\'erieure d'Electronique, d'Electrotechnique, d'Informatique, d'Hydraulique et des T\'el\'ecommunications (ENSEEIHT) de Toulouse in 1989 and the Ph.D. degree from the National Polytechnic Institute from Toulouse in 1992. He is currently a professor in the university of Toulouse (ENSEEIHT) and a member of the IRIT laboratory (UMR 5505 of the CNRS). His research activities are centered around statistical signal and image processing with a particular interest to Bayesian and Markov chain Monte Carlo (MCMC) methods. He has been involved in the organization of several conferences including the European conference on signal processing EUSIPCO'02 (program chair), the international conference ICASSP'06 (plenaries), the statistical signal processing workshop SSP'12 (international liaisons), the International Workshop on Computational Advances in Multi-Sensor Adaptive Processing CAMSAP 2013 (local arrangements), the statistical signal processing workshop SSP'2014 (special sessions), the workshop on machine learning for signal processing MLSP'2014 (special sessions). He has been the general chair of the CIMI workshop on optimization and statistics in image processing hold in Toulouse in 2013 (with F. Malgouyres and D. Kouam\'e) and of the International Workshop on Computational Advances in Multi-Sensor Adaptive Processing CAMSAP 2015 (with P. Djuric). He has been a member of different technical committees including the Signal Processing Theory and Methods (SPTM) committee of the IEEE Signal Processing Society (2001-2007, 2010-present). He has been serving as an associate editor for the IEEE Transactions on Signal Processing (2008-2011, 2015-present) and for the EURASIP journal on Signal Processing (2013-present).
\end{IEEEbiographynophoto}

\end{document}